\documentclass[draftcls,onecolumn,11pt]{IEEEtran}

\usepackage{cite}      

\usepackage{graphicx}  

\usepackage{psfrag}    

\usepackage{subfigure} 

\usepackage{url}       

\usepackage{amssymb}
\usepackage{amsmath}   
\usepackage{array}
\newtheorem{theorem}{Theorem}
\newtheorem{corollary}{Corollary}
\newtheorem{lemma}{Lemma}

\newtheorem{proposition}{Proposition}
\newtheorem{remark}{Remark}

\newcommand{\rspace}{\vspace{-1.7\baselineskip}\\}

\usepackage{xspace}

\begin{document}
%
\title{$n$-Channel Asymmetric Entropy-Constrained Multiple-Description Lattice Vector Quantization}
%
%
\author{Jan~\O stergaard,~\IEEEmembership{Member,~IEEE,}
        Richard~Heusdens, and Jesper~Jensen
\thanks{Manuscript received April 25, 2006; revised August 09, 2010.
This research was performed while all authors were at Delft University of Technology, Delft, The Netherlands,
and was supported by the Technology Foundation STW, applied science division of NWO and the technology programme of the ministry of Economics Affairs. 
Jan {\O}stergaard (janoe@ieee.org) is now with Aalborg University, Aalborg, Denmark. 
Richard Heusdens (r.heusdens@tudelft.nl) is with Delft University of Technology, Delft, The Netherlands,
and Jesper Jensen (jsj@oticon.dk) is now with Oticon, Copenhagen, Denmark. 
This work was presented in part at the International Symposium on Information Theory, 2005 and 2006.}}
%
%
%
%
\markboth{IEEE Transactions on Information Theory, December 2010}{\O stergaard \MakeLowercase{\textit{et al.}}: $n$-Channel Asymmetric Entropy-Constrained Multiple-Description Lattice Vector Quantization}
%
%



\maketitle

\begin{abstract}
This paper is about the design and analysis of an index-assignment (IA) based multiple-description  coding scheme for the $n$-channel asymmetric case. 
We use entropy constrained lattice vector quantization and 
restrict attention to simple reconstruction functions, which are given by the inverse IA function when all descriptions are received or otherwise by a weighted average of the received descriptions.
We consider smooth sources with finite differential entropy rate and MSE fidelity criterion. 
As in previous designs, our construction is based on nested lattices which are combined through a single IA function. 
The results are exact under high-resolution conditions and asymptotically as the nesting ratios of the lattices approach infinity.
For any $n$, the design is asymptotically optimal within the class of
IA-based schemes. 
Moreover, in the case of two descriptions and finite lattice vector dimensions greater than one, the performance is strictly better than that of existing designs. 
In the case of three descriptions, we show that in the limit of large lattice vector dimensions, points on the inner bound of Pradhan et al.\ can be achieved. 
Furthermore, for three descriptions and finite lattice vector
dimensions, we show that the IA-based approach yields, in the
symmetric case,  a smaller rate loss than the recently proposed source-splitting approach. 
\end{abstract}


%
\IEEEpeerreviewmaketitle

\section{Introduction}
%
%
%
%
\label{sec:intro}
\PARstart{M}{ultiple-description} coding (MDC) is about (lossy) encoding of information for transmission over an unreliable $n$-channel communication system. The channels may break down resulting in erasures and a loss of information at the receiving side. 
The receiver knows which subset of the $n$ channels that are working; the transmitter does not. 
The problem is then to design an MDC system which, for given channel rates, minimizes the distortions due to reconstruction of the source using information from any subsets of the channels. 

The achievable multiple-description (MD) rate-distortion function is completely known for the case of two channels, squared-error fidelity criterion and the memoryless Gaussian source~\cite{ozarow:1980,elgamal:1982}. An extension to colored Gaussian sources was provided in~\cite{zamir:1999,zamir:2000,chen:2007}. 
Inner and outer bounds to the $n$-channel quadratic Gaussian rate-distortion region for memoryless sources was presented in~\cite{venkataramani:2003,pradhan:2004,puri:2005,wang:2006,wang:2008,tian:2008b}.

Practical symmetric multiple-description lattice vector quantization (MD-LVQ) based schemes for two descriptions have been introduced in~\cite{vaishampayan:2001}, which in the limit of infinite-dimensional source vectors and under high-resolution assumptions, approach the symmetric MD rate-distortion bound.\footnote{The term symmetric relates to the situation where all channel rates (description rates) are equal and the distortion depends only upon the number of working channels (received descriptions) and as such not on which of the channels that are working. In the asymmetric case, the description rates and side distortions are allowed to be unequal.}
An extension to $n>2$ descriptions was presented in~\cite{ostergaard:2004b}. 
Asymmetric MD-LVQ was considered in~\cite{diggavi:2002} for the case of two descriptions.
Common for all of the designs~\cite{vaishampayan:2001,ostergaard:2004b,diggavi:2002} is that a central quantizer is first applied on the source after which an \emph{index-assignment} (IA) algorithm (also known as a labeling function) maps the reconstruction points of the central quantizer to reconstruction points of the side quantizers, which is an idea that was first presented in~\cite{vaishampayan:1993}. These designs are usually referred to as IA based designs.

There also exists non IA based $n$-channel schemes, which are proven
optimal in the two-channel quadratic Gaussian case. In particular, the
source-splitting approach of Chen et al.~\cite{chen:2006} and the
delta-sigma quantization approach of {\O}stergaard et
al.~\cite{ostergaard:2009,kochman:2008}. In addition, the following
recent work~\cite{mohajer:2008a,tian:2009a}, provide simple
constructions that are approximately optimal.\footnote{Note that the
  recent
  works~\cite{ostergaard:2009,kochman:2008,mohajer:2008a,tian:2009a}
  appeared after the first submission of this paper.}

While the different designs mentioned above are able to achieve the rate-distortion bounds in the asymptotical limit as the lattice vector quantizer dimension $(L)$ gets arbitrarily large, there is an inherent \emph{rate loss} when finite dimensional vector quantizers are employed.\footnote{The term \emph{rate loss} refers to the excess rate due to using a suboptimal implementation.}
For example, in the two-channel symmetric case and at high resolutions, the rate loss (per description) of the IA based schemes is given by $\frac{1}{4}\log_2(G(\Lambda^{(L)})G(S_L) (2\pi e)^2)$ where $G(\Lambda^{(L)})$ is the dimensionless normalized second moment of the $L$-dimensional lattice $\Lambda^{(L)}$ and $G(S_L)$ is the dimensionless normalized second moment of an $L$-dimensional hypersphere~\cite{conway:1999}. 
For the source-splitting approach the rate loss is $\frac{1}{4}\log_2(G(\Lambda^{(L)})^3 (2\pi e)^3)$ whereas for the delta-sigma quantization approach the rate loss is $\frac{1}{4}\log_2(G(\Lambda^{(L)})^2 (2\pi e)^2)$. Since $G(S_L)\leq G(\Lambda^{(L)}), \forall L>0$, it follows that the IA based approaches yield the smallest rate loss of all existing asymptotically optimal designs.\footnote{By use of time-sharing, the rate loss of the source-splitting scheme can be reduced to that of the delta-sigma quantization scheme. Moreover, in the scalar case, the rate loss can be further reduced, see~\cite{chen:2006} for details.}

We will like to point out that there exist a substantial amount of different practical approaches to MDC. For example,  
the work of~\cite{fleming:2004} on asymmetric vector quantization, the work of~\cite{berger-wolf:2002,tian:2004} on $n$-channel scalar quantization and the transform based MDC approaches presented in~\cite{orchard:1997,goyal:2001,balan:2000,chou:1999}. 

In this paper, we are interested in IA based MDC. Specifically, we propose a design of an asymmetric IA based MD-LVQ scheme for the case of $n\geq 2$ descriptions. The design uses a single labeling function and simple reconstruction functions, which are given by the inverse IA function when all descriptions are received or 
otherwise by a weighted average of the received descriptions. 
We consider the case of MSE distortion and smooth sources with finite differential entropy rate.\footnote{For each side description, we assume that the sequence of quantized source vectors is jointly entropy coded using an arbitrarily complex entropy coder.}
To the best of the authors knowledge, the above restrictions (or even less general restriction) are also necessary  for the existing IA-based designs proposed in the literature. 

The contributions of the paper are summarized below and are valid under high-resolution conditions and asymptotically large nesting ratios:
\begin{itemize}
\item We provide a simple construction of the labeling function for
  the asymmetric case and for any number  $n\geq 2$ of
  descriptions. The construction is optimal within the framework of IA
  based schemes where only a single IA function is used and where the 
  reconstruction rule is given as the average of the received
  descriptions (or the central lattice in case all descriptions are
  received).

\item For $n=3$ and any $L\geq 1$, we provide closed-form rate-distortion expressions.

\item For $n=3$ and in the limit as $L\rightarrow\infty$, the distortion points of our scheme lie on the inner bound provided by Pradhan et al.~\cite{pradhan:2004,puri:2005}.

\item For $n=2$ and any $1<L<\infty$, we strictly improve the side
  distortions (by a constant) over that of the asymmetric design by Diggavi et al.~\cite{diggavi:2002}.

\item For $n=3$ and $1\leq L<\infty$, we define a notion of \emph{rate
    loss}  (in the symmetric case only) as
 the difference between the operational rate of the scheme
  and the rate of the inner bound of Pradhan et
  al. ~\cite{pradhan:2004,puri:2005}. We then 
  show that our construction yields a smaller rate loss than that of source-splitting~\cite{chen:2006}.
\end{itemize}

This paper is organized as follows. In Section~\ref{sec:prelim} we briefly review some lattice properties, describe the required asymptotical conditions which we will be assuming through-out the work, and introduce the concept of an IA function. 
The actual design of the MD-LVQ system, which is the main contribution of the paper, is presented in Section~\ref{sec:label}. 
In Section~\ref{sec:comparison}, we compare the proposed design to known inner bounds and existing MD schemes.
The conclusions follow in Section~\ref{sec:conclusion} and appendices are reserved for lengthy proofs.

\section{Preliminaries}\label{sec:prelim}
\subsection{Lattice Properties}
Let the $L$-dimensional real lattice $\Lambda\subset\mathbb{R}^L$ form the codewords of the lattice vector quantizer $\mathcal{Q}_\Lambda(\cdot)$ having Voronoi cells. 
Thus, $\mathcal{Q}_\Lambda(x)=\lambda$ if $x\in V(\lambda)$ where 
$V(\lambda) \triangleq \{ x\in \mathbb{R}^L : \| x - \lambda\|^2 \leq \|x-\lambda'\|^2,\, \forall\, \lambda' \in \Lambda \}$ is a Voronoi cell. We define $\langle x, x \rangle \triangleq \frac{1}{L}x^tx$ and use $\|x\|^2=\langle x, x \rangle$.
The dimensionless normalized second-moment of inertia $G(\Lambda)$ of $\Lambda$ is defined as~\cite{conway:1999}
\begin{equation}\label{eq:G}
G(\Lambda) \triangleq\frac{1}{\nu^{1+2/L}}\int_{V(0)}\|x\|^2dx
\end{equation}
where $V(0)$ is the Voronoi cell around the origin and $\nu$ denotes the volume of $V(0)$. Recall that $\frac{1}{12}\geq G(\Lambda)\geq G(S_L)\geq \frac{1}{2\pi e}$ where $G(S_L) =\frac{1}{(L+2)\pi}\Gamma\left(\frac{L}2 +1\right)^{2/L}$ is the dimensionless normalized second moment of an $L$-dimensional hypersphere
and $\Gamma(\cdot)$ is the Gamma function~\cite{conway:1999}.

Let $\Lambda$ be a lattice, then a sublattice $\Lambda_s \subseteq \Lambda$ is a subset of the elements of $\Lambda$ that is itself a lattice. We say that $\Lambda_s$ is a coarse lattice nested within the fine lattice $\Lambda$. Let $\nu$ and $\nu_s$ be the volumes of $V(0)$ and $V_s(0)$, respectively, where the subscript $s$ indicates the sublattice. Then the index value $N_s$ of $\Lambda_s$ with respect to $\Lambda$ is $N_s = \nu_s/\nu$ and the nesting ratio $N_s'$ is given by $N_s' = \sqrt[L]{N_s}$.

Let $\{\Lambda^{(L)}\}$ be a sequence of lattices indexed by their dimension $L$. Then, $\Lambda^{(L)}$ is said to be asymptotically good for quantization (under MSE) if and only if for any $\epsilon > 0$ and sufficiently large $L$ ~\cite{zamir:2002}
\begin{equation}\label{eq:goodforquantization}
\log_2(2\pi e G(\Lambda^{(L)})) < \epsilon.
\end{equation}

\subsection{The Existence of Lattices and Sublattices for  MD coding}\label{sec:existence}
We need a central lattice (central quantizer) $\Lambda_c$ with Voronoi cell $V_c(0)$ of volume $\nu_c$ and $n$ sublattices (side quantizers) $\Lambda_i\subset \Lambda_c$ with Voronoi cells $V_i(0)$ of volumes $\nu_i$, where $i=0,\dots, n-1$.
Finally, we need a sublattice $\Lambda_\pi \subset \Lambda_i$ which we will refer to as a product lattice. The Voronoi cell $V_\pi(0)$ of $\Lambda_\pi$ has volume $\nu_\pi = N_\pi\nu_c$ where $N_\pi$ is the index value of $\Lambda_\pi$ with respect to $\Lambda_c$. 

Previous work on two-description IA based MD coding focused on the existence and construction of nested lattices for a few low dimensional (root) lattices cf.~\cite{vaishampayan:2001,diggavi:2002}. The techniques of~\cite{vaishampayan:2001,diggavi:2002} was extended to the case of $n$ descriptions for the symmetric case in~\cite{ostergaard:2004b}. While some of the root lattices are considered to be among the best of all lattices (of the same dimensions) for quantization, they are not good for quantization in the sense of~(\ref{eq:goodforquantization}). 
Furthermore, their index values belong to some discrete sets of integers and since they are finite dimensional, arbitrary nesting ratios cannot be achieved. 

Let us first clarify the requirements of the lattices to be used in this work:
\begin{enumerate}
\item The central lattice $\Lambda_c\in \mathbb{R}^L$, is asymptotically good for quantization as $L\rightarrow \infty$.
\item The central lattice $\Lambda_c\in \mathbb{R}^L$ admits sublattices $\Lambda_i \subset \Lambda_c$ of arbitrary nesting ratios $1\leq N_i'\in \mathbb{R}$.
\item There exists a product lattice $\Lambda_\pi\subset\Lambda_i, i=0,\dotsc,n-1,$ with arbitrary nesting ratio $N'_\pi$ (with respect to $\Lambda_c$) where $N'_i< N'_\pi \in \mathbb{R}$ for all $i=0,\dotsc, n-1$.
\end{enumerate}

That there exists a sequence of lattices which are asymptotically good for quantization was established in~\cite{zamir:1996}. It is also known that there exists nested lattices $\Lambda^{(L)}\subset\Lambda_c^{(L)}$ where the coarse lattice ($\Lambda^{(L)}$) is asymptotically good for quantization and the fine lattice ($\Lambda_c^{(L)}$) is asymptotically good for channel coding~\cite{erez:2004}. Moreover, in recent work~\cite{krithivasan:2008}, it has been established that there exists a sequence of nested lattices where the coarse lattice as well as the fine lattice are simultaneously good for quantization. 

Interestingly, we do not require $\{\Lambda_i\}_{i=0}^{n-1}$ nor $\Lambda_\pi$ to be good for quantization. 
This is because we are able to construct a labeling function which, asymptotically as $N_i\rightarrow\infty, \forall i$, results in a distortion that becomes independent of the type of sublattices being used. 
Furthermore, $\Lambda_\pi$ is used to provide a simple construction of a shift invariant region $V_\pi(0)$ and its quantization performance is therefore irrelevant.

We have yet to show the existence of $\Lambda_\pi^{(L)} \subset \Lambda_i^{(L)}$ for $i=0,\dotsc,n-1$. 
Towards that end, we refer to the construction of nested lattices provided in~\cite{krithivasan:2008}. Here a coarse lattice $\Lambda_s^{(L)}$ is first fixed and then a fine lattice $\Lambda_c^{(L)}$ is constructed such that $\Lambda_s^{(L)}\subseteq \Lambda_c^{(L)}$ with an arbitrary nesting ratio. 
Without loss of generality, let $N'_0\leq N'_1\leq \cdots \leq N'_{n-1} < N'_\pi$. 
Moreover, let the set of integers $\mathbb{Z}^L$ form a product lattice $\Lambda_\pi^{(L)}$. Now let $\Lambda_\pi^{(L)}$ be the coarse lattice and construct a fine lattice $\Lambda_{n-1}^{(L)}$ so that the nesting ratio is $N'_\pi/N'_{n-1}$ by using the method of~\cite{krithivasan:2008}. Next, let $\Lambda_{n-1}^{(L)}$ be the coarse lattice and construct a fine lattice $\Lambda_{n-2}^{(L)}$ with a nesting ratio of $N'_{n-1}/N'_{n-2}$. 
This procedure is repeated until the sublattice $\Lambda_0^{(L)}$ is constructed as the fine lattice of $\Lambda_{1}^{(L)}$. At this point, the central lattice $\Lambda_c^{(L)}$ is finally constructed by using $\Lambda_0^{(L)}$ as the coarse lattice and making sure that the nesting ratio is $N'_0$. 
It should be clear that we end up with a sequence of nested lattices, i.e.\ $\Lambda_\pi^{(L)} \subset \Lambda_{n-1}^{(L)}\subseteq \cdots \subseteq \Lambda_0^{(L)} \subset \Lambda_c^{(L)}$ with the desired nesting ratios with respect to $\Lambda_c$, i.e.\ $N'_\pi,N'_{n-1},\cdots,N'_0$.
Without loss of generality, we can take $N'_\pi = \prod_{i=0}^{n-1}N'_i$.\footnote{If $1<m<n$ nesting ratios are identical, we keep only one of them when forming the product lattice. If all nesting ratios are identical, we form the product lattice based on the product of any two of them, see~\cite{ostergaard:2004b} for details.} 
 
In the limit as $L\rightarrow\infty$ it is guaranteed that $\Lambda_c^{(L)}$ becomes asymptotically good for quantization. Furthermore, the sublattices $\Lambda_i^{(L)}$ can be shaped so that they are also good for quantization or they can, for example, be shaped like the cubic lattice. This is not important for the design proposed in this work. 

\subsection{Lattice Asymptotics}
As is common in IA based MD-LVQ, we will in this work require a number of asymptotical conditions to be satisfied in order to guarantee the prescribed rate-distortion performance. Specifically, we require high-resolution conditions, i.e.\ we will be working near the limit where the rates of the central and side quantizers diverge towards infinity, or equivalently, in the limit where the volumes of the Voronoi cells of the lattices in question become asymptotically small. 
This condition makes it possible to assume an approximately uniform source distribution over small regions so that standard high-resolution lattice quantization results become valid~\cite{gray:1990}. Let $\Lambda \subset \mathbb{R}^L$ be a real lattice and let $\nu=\det(\Lambda)$ be the volume of a fundamental region of $\Lambda$. Moreover, let $\tilde{V}\subset\mathbb{R}^L$ be a connected region of volume $\tilde{\nu}$. Then, the high-resolution assumption makes it possible to approximate the number of lattice points in $\tilde{V}$ by $\tilde{\nu}/\nu$, which is an approximation that becomes exact as the number of lattice shells within $\tilde{V}$ goes to infinity.
To be more specific, let $S(c,r)$ be a sphere in $\mathbb{R}^L$ of radius $r$ and center $c\in \mathbb{R}^L$. Then, according to Gauss' counting principle, the number $A_{\mathbb{Z}}$ of integer lattice points in a convex body $\mathcal{C}$ in $\mathbb{R}^L$ equals the volume Vol$(\mathcal{C})$ of $\mathcal{C}$ with a small error term~\cite{mazo:1990}. In fact if $\mathcal{C}=S(c,r)$ then by use of a theorem due to Minkowski it can be shown that, for any $c\in\mathbb{R}^L$ and asymptotically as $r\rightarrow \infty$, $A_{\mathbb{Z}}(r)=\text{Vol}(S(c,r))=\omega_L r^L$, where $\omega_L$ is the volume of the $L$-dimensional unit sphere~\cite{fricker:1982}.
It is also known that the number of lattice points $A_\Lambda(j)$ in the first $j$ shells (i.e., the $j$ shells nearest the origin) of the lattice $\Lambda$ satisfies, asymptotically as $j\rightarrow \infty$, $A_\Lambda(j) = \omega_L j^{L/2}/\nu$~\cite{vaishampayan:2001}. 

In addition to the high-resolution assumption, we also require that the index values of the sublattices become asymptotically large. With this, it follows that the number of central lattice points within a Voronoi cell of a sublattice becomes arbitrarily large. 
Furthermore, to guarantee that the sublattices satisfy the high-resolution quantization properties, we must force the volume of their Voronoi cells to be small. 
In other words, we require that $N_i\rightarrow\infty$ and $\nu_i\rightarrow 0$ where $\nu_i= \nu N_i$ is the volume of a Voronoi cell of the $i$th sublattice. We also note that, in order to avoid that some subset of the sublattices asymptotically dominate the overall distortion, we will require that their index values grow at the same rate, i.e.\ $N_i/N_j=c_{i,j}$ for some constant $c_{i,j}\in\mathbb{R}$.

Finally, as mentioned in the previous section, we require the existence of good lattices for quantization. We therefore require that the lattice vector dimension $L$ tends towards infinity. 

We note that the above asymptotical conditions are only required to guarantee exact results. In fact, at some point, we relax the requirement on $L$
and provide exact results for any $L\geq 1$.
Moreover, the proof technique is constructive in the sense that in non-asymptotical situations, i.e.\ for finite $N_i$ and $R_i$, the results are approximately true. This is interesting from a practical perspective, since, in practice, the asymptotical conditions will never be truly satisfied.

\subsection{Index Assignments}\label{sec:index}
In the MDC scheme considered in this paper, a source vector $x$ is quantized to the nearest reconstruction point $\lambda_c$ in the central lattice $\Lambda_c$. Hereafter follows IAs (mappings), which uniquely map all $\lambda_c$'s to reproduction points in each of the sublattices $\Lambda_i$. This mapping is done through a labeling function $\alpha$, and we denote the individual component functions of $\alpha$ by $\alpha_i$. In other words, the function $\alpha$ that maps $\Lambda_c$ into $\Lambda_0 \times \dots \times \Lambda_{n-1}$, is given by
$\alpha(\lambda_c)=(\alpha_0(\lambda_c),\alpha_1(\lambda_c),\dots,\alpha_{n-1}(\lambda_c)) 
= (\lambda_0,\lambda_1,\dots,\lambda_{n-1})$,
where $\alpha_i(\lambda_c)=\lambda_i \in \Lambda_i$ and $i=0,\dots, n-1$. Each $n$-tuple $(\lambda_0,\dots,\lambda_{n-1})$ is used only once when labeling points in $\Lambda_c$ so that $\lambda_c$ can be recovered unambiguously when all $n$ descriptions are received. 

Since lattices are infinite arrays of points, we adopt the procedure first used in~\cite{vaishampayan:2001} and construct a shift invariant labeling function, so we only need to label a finite number of points. We generalize the approach of~\cite{diggavi:2002} and construct a product lattice $\Lambda_\pi$ which has $N_\pi$ central lattice points and $N_\pi/N_i$ sublattice points from the $i$th sublattice in each of its Voronoi cells. The Voronoi cells $V_\pi(\lambda_\pi)$ of the product lattice $\Lambda_\pi$ are all similar so by concentrating on labeling only central lattice points within one Voronoi cell of $\Lambda_\pi$, the rest of the central lattice points may be labeled simply by translating this Voronoi cell throughout $\mathbb{R}^L$. 
We will therefore only label central lattice points within $V_\pi(0)$, which is the Voronoi cell of $\Lambda_\pi$ around the origin. With this we get 
\begin{equation}\label{eq:shiftinv}
\alpha(\lambda_c + \lambda_\pi) = \alpha(\lambda_c) + \lambda_\pi
\end{equation}
for all $\lambda_\pi \in \Lambda_\pi$ and all $\lambda_c \in \Lambda_c$.

\section{Construction of the Labeling Function}
\label{sec:label}
This section focuses on the labeling problem and is split into several subsections. We begin by Section~\ref{sec:shiftinvariance} which shows how to guarantee shift invariance of the labeling function. 
Then, in Section~\ref{sec:costfunction}, we define the cost function to be minimized by an optimal labeling function. In Section~\ref{sec:construct_ntuples} we show how to construct an optimal set of $n$-tuples and the assignment of the $n$-tuples to central lattice points follows Section~\ref{sec:assignment}. We end by assessing the rate and distortion performances of the labeling function in Section~\ref{sec:rates} and Section~\ref{sec:distortions}, respectively.

\subsection{Guaranteeing Shift Invariance of the Labeling Function}\label{sec:shiftinvariance}
In order to ensure that $\alpha$ is shift-invariant, we must make sure that an $n$-tuple is not assigned to more than one central lattice point $\lambda_c\in \Lambda_c$. 
Notice that two $n$-tuples which are translates of each other by some $\lambda_\pi \in \Lambda_\pi$ must not both be assigned to central lattice points located within the same region $V_\pi(\lambda_\pi)$, since this causes assignment of an $n$-tuple to multiple central lattice points.

The region $V_\pi(0)$ will be translated through-out $\mathbb{R}^L$ and centered at $\lambda_\pi \in \Lambda_\pi$. Assume that $\Lambda_\pi$ is clean\footnote{A sublattice $\Lambda_s\subset \Lambda$ is said to be clean with respect to $\Lambda$ if no points of $\Lambda$ falls on the boundary of the Voronoi cells of $\Lambda_s$. In other words, the set $\{ \lambda \in \Lambda : \lambda \in V_s(\lambda_s) \cap V_s(\lambda'_s)\}$ is empty for all $\lambda_s\neq \lambda_s' \in \Lambda_s$. 
We note that it is an open problem to construct a sequence of nested lattices which are asymptotically good for quantization and where the coarse lattice is clean.}
with respect to $\Lambda_0$.
Then no points of $\Lambda_0$ will be inside more than one $V_\pi(\lambda_\pi)$ region. 
This is the key insight required to guarantee shift invariance. Let us now construct an $n$-tuple, say $(\lambda_0,\lambda_1,\dotsc,\lambda_{n-1})$, where the first element is inside $V_\pi(0)$, i.e.\ $\lambda_0 \in V_\pi(0)$. Once we shift the $n$-tuple by a multiple of $\Lambda_\pi$, the first element of the shifted $n$-tuple will never be inside $V_\pi(0)$ and the $n$-tuple is therefore shift invariant. 
In other words, $(\lambda_0+\lambda_\pi) \notin V_\pi(0)$ for $0\neq\lambda_\pi \in \Lambda_\pi$. 
This also means
that all $n$-tuples (for $\lambda_c\in V_\pi(0)$) have their first
element (i.e.\ $\lambda_0$) inside $V_\pi(0)$. This restriction is
easily removed by considering all cosets of each $n$-tuple. 
Let us define the coset of an $n$-tuple modulo $\Lambda_\pi$ to be 
\begin{equation}\label{eq:cosets}
\begin{split}
&\bar{\mathcal{C}}_{\Lambda_\pi}(\lambda_0,\dots,\lambda_{n-1})
\triangleq \\ 
&\!\{\! (\lambda'_0,\dots,\lambda'_{n-1})\! \in\! \Lambda_0\!\times \cdots \times \!\Lambda_{n-1} : 
\lambda_i' = \lambda_i + \lambda_\pi, \lambda_\pi \in \Lambda_\pi\!\}. 
\end{split}
\end{equation}
The $n$-tuples in a coset are equivalent modulo $\Lambda_\pi$. So by allowing
only one member from each coset to be used when assigning $n$-tuples
to central lattice points within $V_\pi(0)$, the shift-invariance
property is preserved.\footnote{If $\Lambda_\pi$  is not clean, 
  a similar coset construction may be used to systematically deal with the
  boundary points: First, all
  boundary points which are equivalent modulo $\Lambda_\pi$ are within
  the same coset. Second, only one member from each coset is assigned to
  central lattice points.}


\subsection{Defining the Cost Function for the Labeling Problem}\label{sec:costfunction}
We will treat the asymmetric problem where the individual descriptions are weighted and the distortions due to reception of subsets of descriptions are also weighted. 
There are in general several ways of receiving $\kappa$ out of $n$ descriptions. Let $\mathcal{L}^{(n,\kappa)}$ denote an index set consisting of all possible $\kappa$ combinations out of $\{0,\dots, n-1\}$ so that $|\mathcal{L}^{(n,\kappa)}| = \binom{n}{\kappa}$. 
For example, for $n=3$ and $\kappa=2$ we have $\mathcal{L}^{(3,2)}=\{ \{0,1\}, \{0,2\}, \{1,2\}\}$.
Furthermore, let $0 < \mu_i\in \mathbb{R}$ be the weight for the $i$th description.

Recall that $\alpha$ takes a single vector $\lambda_c$ and maps it to a set of vectors $\{\lambda_i\}, i=0,\dots, n-1$, where $\lambda_i\in \Lambda_i$. The mapping is invertible so that we have $\lambda_c=\alpha^{-1}(\lambda_0,\dots,\lambda_{n-1})$. 
Thus, if all $n$ descriptions are received we reconstruct using the inverse map $\alpha^{-1}$ and obtain $\lambda_c$. If no descriptions are received, we reconstruct using the statistical mean of the source. In all other cases, we reconstruct using a weighted average of the received elements. 

We define the reconstruction formula when receiving the set of $\kappa$ out of $n$ descriptions indexed by $\ell\in \mathcal{L}^{(n,\kappa)}$ to be
\begin{equation}\label{rec:formula}
\hat{x}_\ell \triangleq \frac{1}{\kappa}\sum_{i\in \ell} \mu_i\alpha_i(\lambda_c)
\end{equation}
where $1\leq \kappa<n$ and where $\lambda_c = \mathcal{Q}_{\lambda_c}(x)$, i.e.\ $x$ is quantized to $\lambda_c\in \Lambda_c$.
The distortion $d_\ell$ due to approximating $x$ by $\hat{x}_\ell$ is then given by
\begin{equation}
d_{\ell} = \left\|x - \frac{1}{\kappa}\sum_{i\in \ell} \mu_i\alpha_i(\lambda_c)\right\|^2
\end{equation}
and the expected distortion with respect to $X$ is given by $\bar{D}_\ell = \mathbb{E}d_\ell$.
\begin{lemma}[\cite{vaishampayan:2001}]\label{lem:expdist_split}
For any $1\leq \kappa<n$, $\ell\in\mathcal{L}^{(n,\kappa)}$, asymptotically as $\nu_c\rightarrow 0$ and independently of $\alpha$
\begin{equation}
\begin{split}
\bar{D}_\ell
&=\sum_{\lambda_c\in\Lambda_c}\int_{V_c(\lambda_c)}f_X(x)\left\|X-\lambda_c\right\|^2\,dx
\\
&\quad+ \sum_{\lambda_c\in\Lambda_c}\int_{V_c(\lambda_c)}f_X(X)
\left\|\lambda_c - \frac{1}{\kappa}\sum_{i\in \ell} \mu_i\alpha_i(\lambda_c)\right\|^2\,dx. \label{eq:Dbarl}
\end{split}
\end{equation}
\end{lemma}
\begin{proof}
The lemma was proved in~\cite{vaishampayan:2001} for the symmetric case and two descriptions. The extension to the asymmetric case and $n$ descriptions is straight forward. See~\cite{ostergaard:2007a} for details.
\end{proof}

Notice that only the second term of~(\ref{eq:Dbarl}) is affected by the labeling function. We will make use of this fact and therefore define
\begin{equation}
D_\ell \triangleq \left\|\lambda_c - \frac{1}{\kappa}\sum_{i\in \ell} \mu_i\alpha_i(\lambda_c)\right\|^2 \label{eq:Dl}.
\end{equation}

The cost function to be minimized by the labeling function must take
into account the entire set of distortions due to reconstructing from
different subsets of descriptions. With this in mind, we combine the
distortions through a set of scalar (Lagrangian)
weights. Specifically, let $\gamma_\ell\in\mathbb{R}, \ell\in
\mathcal{L}^{(n,\kappa)}$ be the weight for the distortion $D_\ell$
due to reconstructing using the set of descriptions indexed by
$\ell$. With this, we define the cost function $\mathcal{J}^n$ for the
$n$-description labeling problem to be given by~(\ref{eq:costlarge})
(see top of next page),
\begin{figure*}
\begin{equation}\label{eq:costlarge}
\begin{split}
\displaystyle\mathcal{J}^n &\triangleq \displaystyle\sum_{\lambda_c\in\Lambda_c} \int_{V_c(\lambda_c)}f_X(x)
\Bigg\{\sum_{i=0}^{n-1} \gamma_i\bigg\|\lambda_c -
\mu_i\alpha_i(\lambda_c)\bigg\|^2 
+ \sum_{i=0}^{n-2}\sum_{j=i+1}^{n-1} \gamma_{i,j}\bigg\|\lambda_c - \frac{\mu_i\alpha_i(\lambda_c)+ \mu_j\alpha_j(\lambda_c)}{2}\bigg\|^2 \\
&\quad + \sum_{i=0}^{n-3}\sum_{j=i+1}^{n-2}\sum_{k=j+1}^{n-1} \gamma_{i,j,k}\bigg\|\lambda_c - \frac{\mu_i\alpha_i(\lambda_c)+ \mu_j\alpha_j(\lambda_c) + \mu_k\alpha_k(\lambda_c)}{3}\bigg\|^2 + \cdots \Bigg\}\,dx.
\end{split}
\end{equation}
\hrule
\end{figure*}
which can be written more compactly as
\begin{equation}
\mathcal{J}^n \triangleq \sum_{\lambda_c\in \Lambda_c}\int_{V_c(\lambda_c)}f_X(x)\sum_{\kappa=1}^{n-1}\sum_{\ell\in \mathcal{L}^{(n,\kappa)}} \gamma_\ell D_\ell \, dx.
\end{equation}
For example, using the fact that $\lambda_i=\alpha_i(\lambda_c)$, we can write $\mathcal{J}^n$ for the $n=3$ case as
\begin{equation*}
\begin{split}
&\mathcal{J}^3 = \sum_{\lambda_c\in\Lambda_c}\int_{V_c(\lambda_c)}f_X(x)
\sum_{i=0}^{2} \gamma_i\bigg\|\lambda_c - \mu_i\lambda_i\bigg\|^2\, dx \\
&+\!\! \sum_{\lambda_c\in\Lambda_c}\int_{V_c(\lambda_c)}f_X(x)
\sum_{i=0}^{1}\sum_{j=i+1}^{2} \gamma_{i,j}\bigg\|\lambda_c - \frac{\mu_i\lambda_i+ \mu_j\lambda_j}{2}\bigg\|^2\!dx.
\end{split}
\end{equation*}

Since we are considering the high-resolution regime, we can make the following simplifications
\begin{align}
\mathcal{J}^n &= 
\sum_{\lambda_c\in \Lambda_c}\sum_{\kappa=1}^{n-1}\sum_{\ell\in \mathcal{L}^{(n,\kappa)}} \gamma_\ell \int_{V_c(\lambda_c)}f_X(x)D_\ell \, dx\\
&= \sum_{\lambda_c\in \Lambda_c}P(X\in V_c(\lambda_c))\sum_{\kappa=1}^{n-1}\sum_{\ell\in \mathcal{L}^{(n,\kappa)}} \gamma_\ell D_\ell \\
&\approx \sum_{\lambda_\pi\in \Lambda_\pi}\frac{P(X\in V_\pi(\lambda_\pi))}{N_\pi}\sum_{\lambda_c\in V_\pi(\lambda_\pi)}\sum_{\kappa=1}^{n-1}\sum_{\ell\in \mathcal{L}^{(n,\kappa)}} \gamma_\ell D_\ell  \label{eq:shiftinvJn}\\
&=\frac{1}{N_\pi}\sum_{\lambda_c\in V_\pi(0)}\sum_{\kappa=1}^{n-1}\sum_{\ell\in \mathcal{L}^{(n,\kappa)}} \gamma_\ell D_\ell
\end{align}
where $P(X\in V_c(\lambda_c))$ is the probability that $X$ will be mapped (or quantized) to $\lambda_c$. The approximation follows by substituting $P(X\in V_c(\lambda_c))\approx P(X\in V_\pi(\lambda_\pi))/N_\pi$ for $\lambda_\pi\in\Lambda_\pi$ which becomes exact as $\nu_i\rightarrow 0$.
In~(\ref{eq:shiftinvJn}), we also exploited that $\alpha$ is shift invariant in order to decompose the sum $\sum_{\lambda_c\in\Lambda_c}$ into the double sum $\sum_{\lambda_\pi\in \Lambda_\pi}\sum_{\lambda_c\in V_\pi(\lambda_\pi)}$ as follows from~(\ref{eq:shiftinv}).

We would like to simplify $\mathcal{J}^n$ even further. In order to do so, we introduce the following notation. 
Let $\mathcal{L}_{i}^{(n,\kappa)}$ indicate the set of all $\ell\in \mathcal{L}^{(n,\kappa)}$ that contains the index $i$, i.e., $\mathcal{L}_i^{(n,\kappa)}=\{ \ell \in \mathcal{L}^{(n,\kappa)} : i\in \ell\}$.
Similarly, $\mathcal{L}_{i,j}^{(n,\kappa)} = \{ \ell \in \mathcal{L}^{(n,\kappa)} : i,j\in \ell\}$. 
Moreover, let $\bar{\gamma}(\mathcal{L}^{(n,\kappa)}) = \sum_{\ell\in\mathcal{L}^{(n,\kappa)}} \gamma_\ell$, $\bar{\gamma}(\mathcal{L}_i^{(n,\kappa)}) = \sum_{\ell\in\mathcal{L}_i^{(n,\kappa)}} \gamma_\ell$ and $\bar{\gamma}(\mathcal{L}_{i,j}^{(n,\kappa)}) = \sum_{\ell\in\mathcal{L}_{i,j}^{(n,\kappa)}} \gamma_\ell$. 
Thus, $\bar{\gamma}(\mathcal{L}^{(3,2)})=\gamma_{0,1} + \gamma_{0,2} + \gamma_{1,2}$ and $\bar{\gamma}(\mathcal{L}_1^{(3,2)})=\gamma_{0,1}+\gamma_{1,2}$.

\begin{theorem}\label{theo:Jreduced}
Let $1\leq \kappa<n<\infty$. Given a set of distortion weights $\{\gamma_\ell \in \mathbb{R} : \ell\in \mathcal{L}^{(n,\kappa)}, 1\leq \kappa\leq n-1\}$, a set of description weights $\{0< \mu_{i} \in \mathbb{R} : i=0,\dotsc, n-1\}$ and any $\lambda_c\in \Lambda_c$ we have
\begin{equation}\label{eq:Jreduced1}
\begin{split}
&\sum_{\ell\in \mathcal{L}^{(n,\kappa)}}  \gamma_\ell D_\ell 
= \sum_{i=0}^{n-2}\sum_{j=i+1}^{n-1}\hat{\gamma}_{i,j}^{(n,\kappa)}\bigg\|\mu_i\lambda_i - \mu_j\lambda_j\bigg\|^2 \\
&+ \bar{\gamma}(\mathcal{L}^{(n,\kappa)})\bigg\| \lambda_c - \frac{1}{\kappa \bar{\gamma}(\mathcal{L}^{(n,\kappa)})}\sum_{i=0}^{n-1}\bar{\gamma}(\mathcal{L}^{(n,\kappa)}_i) \mu_i\lambda_i\bigg\|^2
\end{split}
\end{equation}
where $\lambda_i=\alpha_i(\lambda_c)$ and
\begin{equation}\label{eq:gamma_ij}
\hat{\gamma}_{i,j}^{(n,\kappa)}=\frac{1}{\kappa^2}\Bigg(\frac{\bar{\gamma}(\mathcal{L}^{(n,\kappa)}_i)\bar{\gamma}(\mathcal{L}^{(n,\kappa)}_j)}{\bar{\gamma}(\mathcal{L}^{(n,\kappa)})}-\bar{\gamma}(\mathcal{L}^{(n,\kappa)}_{i,j})\Bigg).
\end{equation}
\end{theorem}
\begin{IEEEproof}
See Appendix~\ref{app:proof_Jreduced}.
\end{IEEEproof}

From~(\ref{eq:Jreduced1}) we make the observation that whenever $\lambda_i$ appears, it is multiplied by $\mu_i$. Without loss of generality, we can therefore scale the lattice $\Lambda_i$ by $\mu_i$ and consider the scaled lattice $\tilde{\Lambda}_i=\mu_i\Lambda_i$ instead. This simplifies the notation. For example, $\hat{x}_\ell = \frac{1}{\kappa}\sum_{i\in \ell}\tilde{\lambda}_i$ where $\tilde{\lambda}_i = \mu_i\lambda_i$ for $i=0,\dotsc,n-1$. Clearly, scaling the sublattices affects the side description rates. We address this issue in Section~\ref{sec:rates}.

By use of Theorem~\ref{theo:Jreduced} we can rewrite the cost function to be minimized by the labeling function as
\begin{equation}\label{eq:simplecost}
\begin{split}
\mathcal{J}^n &= \frac{1}{N_\pi}
\sum_{\lambda_c\in V_\pi(0)}\sum_{\kappa=1}^{n-1}
\bigg\{
\sum_{i=0}^{n-2}\sum_{j=i+1}^{n-1}\hat{\gamma}_{i,j}^{(n,\kappa)}\bigg\|\tilde{\lambda}_i
- \tilde{\lambda}_j\bigg\|^2 \\
&+ \bar{\gamma}(\mathcal{L}^{(n,\kappa)})\bigg\| \lambda_c - \frac{1}{\kappa \bar{\gamma}(\mathcal{L}^{(n,\kappa)})}\sum_{i=0}^{n-1}\bar{\gamma}(\mathcal{L}^{(n,\kappa)}_i)\tilde{\lambda}_i\bigg\|^2\bigg\}
\end{split}
\end{equation}
where $\hat{\gamma}_{i,j}^{(n,\kappa)}$ is given by~(\ref{eq:gamma_ij}).

The following theorem allows us to simplify the construction of the labeling function:
\begin{theorem}\label{theo:separable}
Let $1<n\in \mathbb{N}$. The cost function $\mathcal{J}^n$ is asymptotically separable in the sense that, as $N_i\rightarrow \infty$ and $\nu_i\rightarrow 0, \forall i$, an optimal set $\mathcal{T}^*$ of $N_\pi$ distinct and shift invariant $n$-tuples satisfies
\begin{equation}\label{eq:pairwise_distances}
\mathcal{T}^* = \arg\min_{\mathcal{T}} \sum_{(\lambda_0,\dotsc,\lambda_{n-1})\in\mathcal{T}}\sum_{\kappa=1}^{n-1}\sum_{i=0}^{n-2}\sum_{j=i+1}^{n-1}\hat{\gamma}_{i,j}^{(n,\kappa)}\bigg\|\tilde{\lambda}_i - \tilde{\lambda}_j\bigg\|^2
\end{equation}
where $\mathcal{T}=\{ (\lambda_0,\dotsc,\lambda_{n-1})\in \Lambda_0\times\cdots\times\Lambda_{n-1} : (\lambda_0,\dotsc,\lambda_{n-1}) \text{ is shift invariant} \}, |\mathcal{T}|=N_\pi$ and where $\hat{\gamma}_{i,j}^{(n,\kappa)}$ is given by~(\ref{eq:gamma_ij}).
\end{theorem}
\begin{IEEEproof}
See Appendix~\ref{app:separable}.
\end{IEEEproof}

Theorem~\ref{theo:separable} provides a guideline for the construction of $n$-tuples. One should first find a set of $N_\pi$ distinct and shift invariant $n$-tuples which satisfies~(\ref{eq:pairwise_distances}). 
These $n$-tuples (or members of their cosets) should then be assigned to central lattice points in $V_\pi(0)$ such that 
\begin{equation}\label{eq:central_assignment}
\sum_{\lambda_c\in V_\pi(0)}\sum_{\kappa=1}^{n-1} \bigg\| \lambda_c - \frac{1}{\kappa \bar{\gamma}(\mathcal{L}^{(n,\kappa)})}\sum_{i=0}^{n-1}\bar{\gamma}(\mathcal{L}^{(n,\kappa)}_i) \tilde{\lambda}_i\bigg\|^2
\end{equation}
is minimized. 

\begin{remark}
Notice that we have not claimed that $\mathcal{T}^*$ is unique. Thus, there might be several sets of $n$-tuples which all satisfy~(\ref{eq:pairwise_distances}) but yield different distortions when inserted in~(\ref{eq:central_assignment}). 
However, Theorem~\ref{theo:separable} states that the asymptotically (as $N_i\rightarrow \infty$) dominating distortion is due to that of~(\ref{eq:pairwise_distances}). Thus, any set of $n$-tuples satisfying~(\ref{eq:pairwise_distances}) will be asymptotically optimal. 
Interestingly, we show in Section~\ref{sec:assignment} that $\mathcal{T}^*$ is, in certain cases, indeed asymptotically unique (up to translations by coset members).
\end{remark}

\subsection{Constructing $n$-Tuples}\label{sec:construct_ntuples}
In order to construct $n$-tuples which are shift invariant we extend the technique previously proposed for the symmetric $n$-description MD problem~\cite{ostergaard:2004b}. 

We first center a sphere $\tilde{V}$ at all sublattice points $\lambda_0 \in V_\pi(0)$ and construct $n$-tuples by combining sublattice points from the other sublattices (i.e.\ $\Lambda_i, i=1,\dots,n-1$) within $\tilde{V}(\lambda_0)$ in all possible ways and select the ones that minimize~(\ref{eq:pairwise_distances}). 
For each $\lambda_0\in V_\pi(0)$ it is possible to construct $\prod_{i=1}^{n-1}\tilde{N}_i$ $n$-tuples, where $\tilde{N}_i$ is the number of sublattice points from the $i$th sublattice within the region $\tilde{V}$. 
This gives a total of $(N_\pi/N_0)\prod_{i=1}^{n-1}\tilde{N}_i$ $n$-tuples when all $\lambda_0\in V_\pi(0)$ are used. 
The number $\tilde{N}_i$ of lattice points within $\tilde{V}$ may be approximated by $\tilde{N}_i\approx \tilde{\nu}/\nu_i$ where $\tilde{\nu}$ is the volume of $\tilde{V}$.\footnote{This approximation becomes exact in the usual asymptotical sense of $N_i\rightarrow \infty$ and $\nu_i\rightarrow 0$.}

Since $\tilde{N}_i\approx\tilde{\nu}/(\nu N_i)$ and we need $N_0$ $n$-tuples for each $\lambda_0\in V_\pi(0)$ we see that
\begin{equation*}
N_0 \leq \prod_{i=1}^{n-1}\tilde{N}_i\approx \frac{\tilde{\nu}^{n-1}}{\nu^{n-1}}\prod_{i=1}^{n-1}N_i^{-1},
\end{equation*}
so in order to obtain at least $N_0$ $n$-tuples the volume of $\tilde{V}$ must satisfy (asymptotically as $N_i\rightarrow \infty$)
\begin{equation}\label{eq:vtilde}
\tilde{\nu}\geq \nu_c\prod_{i=0}^{n-1}N_i^{1/(n-1)}.
\end{equation}
For the symmetric case, i.e.\ $N=N_i$, $i=0,\dots,n-1$, we have $\tilde{\nu}\geq \nu_c N^{n/(n-1)}$, which is in agreement with the results obtained in~\cite{ostergaard:2004b}.

The design procedure can be outlined as follows:
\begin{enumerate}
\item Center a sphere $\tilde{V}$ at each $\lambda_0\in V_\pi(0)$ and construct all possible $n$-tuples $(\lambda_0,\lambda_1,\dots,\lambda_{n-1})$ where $\lambda_i\in \tilde{V}(\lambda_0)$ and $i=1,\dots, n-1$. 
This makes sure that all $n$-tuples have their first element ($\lambda_0$) inside $V_\pi(0)$ and they are therefore shift-invariant. 
\item Keep only $n$-tuples whose elements satisfy $\|\lambda_i-\lambda_j\|^2\leq r^2/L, \forall i,j\in 0,\dots n-1$, where $r$ is the radius of $\tilde{V}$.  \label{enum:small_ntuples}
\item Make $\tilde{V}$ large enough so at least $N_0$ distinct $n$-tuples are found for each $\lambda_0$.
\end{enumerate}

The restriction $\|\lambda_i-\lambda_j\|^2\leq r^2/L$ in
step~\ref{enum:small_ntuples} above, is imposed to avoid bias towards
any of the sublattices. At this point, one might wonder why we wish to
avoid such bias. After all, the expression to be minimized,
i.e.~(\ref{eq:pairwise_distances}), includes weights
$\hat{\gamma}_{i,j}^{(n,\kappa)}$ (which might not be equal) for every
pair of sublattices. In otherwords, why not use spheres
$\tilde{V}_{i,j}$ of different sizes to guarantee that
$\|\lambda_i-\lambda_j\|^2\leq r_{i,j}^2/L$ where the radius $r_{i,j}$
now depends on the particular pair of sublattices under
consideration. This is illustrated in Fig.~\ref{fig:3spheres} in Appendix~\ref{app:separable}, where $r_{i,j}$ denotes the radius of the sphere $\tilde{V}_{i,j}$. Here we center $\tilde{V}_{0,1}$ at some $\lambda_0 \in V_\pi(0)$ as illustrated in Fig.~\ref{fig:3spheres} by the solid circle. Then, for any $n$-tuples having this $\lambda_0$ point as first element, we only include $\lambda_1$ points which are inside $\tilde{V}_{0,1}(\lambda_0)$. This guarantees that $\|\lambda_0 - \lambda_1\|^2 \leq r_{0,1}/L$. 
Let us now center a sphere $\tilde{V}_{1,2}$ at some $\lambda_1$ which is inside $\tilde{V}_{0,1}(\lambda_0)$. This is illustrated by the dotted sphere of radius $r_{1,2}$ in the figure. We then only include $\lambda_2$ points which are in the intersection of $\tilde{V}_{1,2}(\lambda_1)$ and $\tilde{V}_{0,2}(\lambda_0)$. This guarantees that $\|\lambda_i - \lambda_j\|^2 \leq r_{i,j}/L$ for all $(i,j)$ pairs. 

Clearly, the radius $r_{i,j}$ must grow at the same rate for any pair $(i,j)$ so that, without loss of generality, $r_{0,1}=a_2r_{0,2}=a_1r_{1,2}$ for some fixed $a_1,a_2\in \mathbb{R}$. 
Interestingly, from Fig.~\ref{fig:3spheres} we see that $r_{0,2}$
cannot be greater than $r_{0,1} + r_{1,2}$ which effectively upper
bounds $a_2$. Thus, the ratio $r_{i,j}/r_{k,l}$ cannot be
arbitrary. Furthermore, it is important to see that the asymmetry in
distortions between the descriptions, is not dictated by $r_{i,j}$ but
instead by how the $n$-tuples are assigned to the central lattice
points. Recall from~(\ref{eq:central_assignment}) that the assignment
is such that the distances between the central lattice points and the
weighted centroids of the $n$-tuples are minimized. In other words, if
we wish to reduce the distortion due to receiving description $i$,
then we
assign the $n$-tuples so that the $i$th element of the $n$-tuples is closer (on average) to the associated central lattice points. Obviously, the remaining elements of the $n$-tuples will then be further away from the assigned central lattice points. 

In the following we first consider the case where $r=r_{i,j}$ for any $(i,j)$. We later show that this is indeed the optimal choice in the symmetric distortion case. It is trivially also optimal in the two-channel asymmetric case, since there is only a single weight $\hat{\gamma}^{(2,1)}_{0,1}$.
In general, we can always scale the radii such that 
\begin{equation}\label{eq:cij}
\begin{split}
\sum_{(\lambda_0,\dotsc,\lambda_{n-1})\in
  \mathcal{T}} & \sum_{\kappa=1}^{n-1}
\hat{\gamma}_{i,j}^{(n,\kappa)}\|\lambda_i - \lambda_j \|^2  =   \\
&\sum_{(\lambda_0,\dotsc,\lambda_{n-1})\in \mathcal{T}}
c_{k,l}\,\sum_{\kappa=1}^{n-1}
\hat{\gamma}_{k,l}^{(n,\kappa)}\|\lambda_k - \lambda_l \|^2 
\end{split}
\end{equation}
for any $(i,j) \neq (k,l)$ where $\mathcal{T}$ indicates the set of $N_\pi$ $n$-tuples and $c_{k,l}\in\mathbb{R}$. The resulting distortions weights (as given by~(\ref{eq:hatgammak1}) and~(\ref{eq:hatgammak2})) should then include the additional set of scaling factors $\{c_{k,l}\}$. This case is treated by Lemma~\ref{lem:c_ij}.

We now proceed to find the optimal $\tilde{\nu}$, i.e.\ the smallest volume which (asymptotically for large $N_i$) leads to exactly $N_0$ tuples satisfying step~\ref{enum:small_ntuples}. 
In order to do so, we adopt the approach of~\cite{ostergaard:2004b} and introduce a dimensionless expansion factor $\psi_{n,L}$. 
The expansion factor $\psi_{n,L}$ describe how much $\tilde{V}$ needs to be expanded (per dimension) from the theoretical lower bound~(\ref{eq:vtilde}), to make sure that exactly $N_0$ optimal $n$-tuples can be constructed by combining sublattice points within a region $\tilde{V}$. 
With this approach, we have that 
\begin{equation}\label{eq:nutilde}
\tilde{\nu}=\psi_{n,L}^L\nu_c\prod_{i=0}^{n-1}N_i^{1/(n-1)}.
\end{equation}

In practice, it is straight-forward to determine $\psi_{n,L}$. One can simply start at $\psi_{n,L}=1$ and iteratively increase $\psi_{n,L}$ in small steps until exactly $N_0$ $n$-tuples are found which all satisfy $\|\lambda_i - \lambda_j\|^2\leq r/L$. 
For volumes containing a large number of lattice points, i.e.\  asymptotically as $N_i\rightarrow\infty$, such an approach determines $\psi_{n,L}$ to arbitrary accuracy. 
Furthermore, in this asymptotical case, $\psi_{n,L}$ becomes independent of the type of lattice (and also $N_i)$, since it then only depends on the number of lattice points within a large volume. Thus, it should be clear that for any $1<n\in\mathbb{N}$ and $1\leq L\in\mathbb{N}$, and asymptotically as $N_i\rightarrow\infty, \forall i$, there exist a unique $1\leq \psi_{n,L}\in\mathbb{R}$.

In general, it is complicated to find an analytical expression for $\psi_{n,L}$. However, we have previously been able to do it for the symmetric MD problem in some interesting cases. 
It turns out that the proof technique and solutions provided for the symmetric case, carry over to the asymmetric case. To see this, we sketch the proof technique here for the asymmetric case and $n=3$. 

Recall that we seek $3$-tuples such that any two members of the $3$-tuple is distanced no more than $r^2/L$ apart. Specifically, we require $\|\lambda_i-\lambda_j\|^2\leq r^2/L$ where $r$ is the radius of $\tilde{V}$. 
Essentially, this is a counting problem. We first center a sphere $\tilde{V}$ at some $\lambda_0\in V_\pi(0)\cap \Lambda_0$. Then we pick a single $\lambda_1\in \tilde{V}(\lambda_0)\cap \Lambda_1$. Finally, we center an equivalent sphere $\tilde{V}$ at this $\lambda_1$ and count the number, say $\#_{\lambda_1}$, of $\lambda_2\in \tilde{V}(\lambda_0)\cap \tilde{V}(\lambda_1) \cap \Lambda_2$. Thus, there is $\#_{\lambda_1}$ $3$-tuples having the same pair $(\lambda_0,\lambda_1)$ as first and second element.
The procedure is now repeatedly applied for all $\lambda_1\in \tilde{V}(\lambda_0)\cap \Lambda_1$ yielding the total number of $3$-tuples to be $\sum_{\lambda_1\in\tilde{V}(\lambda_0)} \#_{\lambda_1}$ (all having the same $\lambda_0$ as first element). 

For large volumes, the number of lattice points in a region $S$ is given by $\text{Vol}(S)/\nu_2$ where $\text{Vol}(S)$ is the volume of $S$ and $\nu_2$ is the volume of the Voronoi cell of the sublattice points $\lambda_2\in\Lambda_2$. Thus, given the pair $(\lambda_0,\lambda_1)$, the number of $\lambda_2$ sublattice points is approximately given by $\text{Vol}(S)/\nu_2$ where $S=\tilde{V}(\lambda_0)\cap \tilde{V}(\lambda_1)$. 
It follows that we need to find the radius (or actually the volume $\tilde{\nu}$ of $\tilde{V}$) such that 
$\sum_{\lambda_1\in\tilde{V}(\lambda_0)} \#_{\lambda_1} = N_0$, since we need exactly $N_0$ $3$-tuples for each $\lambda_0 \in V_\pi(0)\cap \Lambda_0$. 
To find the optimal $\tilde{\nu}$, we derive the volume of intersecting $L$-dimensional spheres distanced $0<b\in \mathbb{R}$ apart. We then let ${b_k}$ be a sequence of increasing distances which yields a sequence of volumes $\{\text{Vol}(S_k)\}$ of the partial intersections $S_k = \tilde{V}(0)\cap \tilde{V}(b_k)$. We finally form the equality $\sum_{k=1}^{r}\#_{S_k}\text{Vol}(S_k)/\nu_2=N_0$ where $\#_{S_k}$ denotes the number of times each $S_k$ occurs. By solving for $r$, we find the unique volume $\tilde{\nu}$ which leads to exactly $N_0$ $n$-tuples.
It can be shown that this procedure yields the optimal $\tilde{\nu}$ and is asymptotically exact for large volumes. Furthermore, it is essentially equivalent to the symmetric case the only exception being that the index values forming the product~(\ref{eq:nutilde}) are allowed to be different in the asymmetric case.
We therefore refer the reader to~\cite{ostergaard:2004b,ostergaard:2007a} for the rigorous proof and quote some results below.

In the case of $n=2$, it trivially follows that $\psi_{2,L}=1$ for all $L$. For the case of $n=3$ and $L$ odd we have the following theorem.
\begin{theorem}[{~\cite[Thm.~3.2]{ostergaard:2004b}}]
Let $n=3$. Asymptotically as $N_i\rightarrow\infty, \nu_i\rightarrow 0, \forall i$, $\psi_{3,L}$ for $L$ odd is given by
\begin{equation}\label{eq:psi3L}
\psi_{3,L} = \left(\frac{\omega_L}{\omega_{L-1}}\right)^{\frac{1}{2L}}\left(\frac{L+1}{2L}\right)^{\frac{1}{2L}}\beta_{L}^{-\frac{1}{2L}}
\end{equation}
where $\omega_L$ denotes the volume of an $L$-dimensional unit-sphere and $\beta_L$ only depends on $L$ and is given by
\begin{equation}\label{eq:betaL}
\begin{split}
\beta_L &=
\sum_{m=0}^{\frac{L+1}2} \binom{\frac{L+1}2}{m}
2^{\frac{L+1}2-m}(-1)^m \sum_{k=0}^{\frac{L-1}2} \frac{
  (\frac{L+1}2)_k (\frac{1-L}2)_k }{ (\frac{L+3}2)_k\, k! } \\
&\quad \times
\sum_{j=0}^{k}\binom{k}{j}\bigg(\frac{1}{2}\bigg)^{k-j}(-1)^j\bigg(\frac{1}4\bigg)^{j}\frac{1}{L+m+j}
\end{split}
\end{equation}
where $(\cdot)_k$ is the Pochhammer symbol.\footnote{The Pochhammer symbol is defined as $(a)_0=1$ and $(a)_k = a(a+1)\cdots (a+k-1)$ for $k\geq 1$.}
\hfill$\blacktriangle$
\end{theorem}

\begin{theorem}[\cite{ostergaard:2004b,ostergaard:2007a}]
Let $n=3$. Asymptotically as $N_i\rightarrow\infty, \nu_i\rightarrow 0,\forall i$, and $L\rightarrow\infty$
\begin{equation}
\psi_{3,\infty} = \left(\frac{4}{3}\right)^{\frac{1}4}.
\end{equation}
\hfill$\blacktriangle$
\end{theorem}

\begin{remark}
The proposed construction also provides a shift invariant set of $n$-tuples in the non-asymptotical case where $N_i$ is finite. Thus, the design is useful in practice.
\end{remark}

\subsection{Assigning $n$-Tuples to Central Lattice Points}\label{sec:assignment}
At this point, we may assume that we have a set $\mathcal{T}$ containing $N_\pi$ shift invariant $n$-tuples.
These $n$-tuples need to be assigned to the $N_\pi$ central lattice points within $V_\pi(0)$. 
However, before doing so, we first construct the coset of each $n$-tuple of $\mathcal{T}$. 
Recall that the coset of an $n$-tuple is given by~(\ref{eq:cosets}).

As first observed by Diggavi et al.~\cite{diggavi:2002}, assignment of $n$-tuples (or more correctly cosets of $n$-tuples) to central lattice points, is a standard linear assignment problem where only one member from each coset is assigned. This guarantees that the labeling function is shift invariant.
The cost measure to be minimized by the linear assignment problem is given by~(\ref{eq:central_assignment}). Thus, the sum of distances between the weighted centroids of the $n$-tuples and the central lattice points should be minimized.

\begin{remark}
Notice that we have shown that there exists a set of $n$-tuples and an assignment that satisfy the desired set of distortions. 
However, there might exist several assignments (for the same set of $n$-tuples) all yielding the same overall Lagrangian cost. 
Thus, in practice, when solving the bipartite matching problem one might need to search through the complete set of solutions (assignments) in order to find one that leads to the desired set of distortions. 
Alternatively, one can pick different solutions (assignments) and use each of them a certain amount of time so that on average the desired set of distortions are satisfied. 
\end{remark}

\begin{remark}
It might appear that the shift invariance restriction enforced by using only one member from each coset will unfairly penalize $\Lambda_0$. However, the following theorems prove that, asymptotically as $N_i\rightarrow \infty$, there is no bias towards any of the sublattices. 
We will consider here the case of $n>2$ (for $n=2$ we can use similar arguments as given in~\cite{diggavi:2002}).
\end{remark}

\begin{theorem}\label{theo:nobias}
Let $n>2$. Asymptotically as $N_i\rightarrow \infty, \forall i$, the number of $n$-tuples that includes sublattice points outside $V_\pi(0)$ becomes negligible compared to the number of $n$-tuples which have all there sublattice points inside $V_\pi(0)$. 
\end{theorem}
\begin{IEEEproof}
See Appendix~\ref{app:proof_nobias}.
\end{IEEEproof}

\begin{theorem}\label{theo:equivtuples}
Let $n>2$. Asymptotically as $N_i\rightarrow\infty,\forall i$, the set of $N_\pi$ $n$-tuples that is constructed by centering $\tilde{V}$ at each $\lambda_0\in V_\pi(0)\cap \Lambda_0$ becomes identical to the set constructed by centering $\tilde{V}$ at each $\lambda_i\in V_\pi(0)\cap \Lambda_i$, where $i\in \{1,\dots,n-1\}$.
\end{theorem}
\begin{IEEEproof}
See Appendix~\ref{app:proof_equivtuples}.
\end{IEEEproof}

\begin{remark}
Notice that the above theorems imply that the set of $n$-tuples which satisfies~(\ref{eq:pairwise_distances}) and is constructed so that $r_{i,j}=a_{i,j}r, \forall (i,j)$ and $a_{i,c}\in\mathbb{R}^L$, is unique (at least up to translations by members of their cosets). The assignment of the $n$-tuples to central lattice points, however, might not be unique.
\end{remark}

\subsection{Description Rates}\label{sec:rates}
The single-description rate $R_c$, i.e.\ the rate of the central quantizer, is given by
\begin{equation*}
R_c =\! -\frac{1}L \!\sum_{\lambda_c\in \Lambda_c}\!\!\!\bigg(\!\int_{V_c(\lambda_c)}f_X(x)dx\bigg)\,\log_2\bigg(\int_{V_c(\lambda_c)}f_X(x)dx\!\bigg).
\end{equation*}
Using the fact that each Voronoi cell $V_c(\lambda_c)$ has identical volume $\nu_c$ and assuming that $f_X(x)$ is approximately constant over Voronoi cells of the central lattice $\Lambda_c$, it can be shown that~\cite{gray:1990}
\begin{equation}\label{eq:Rc}
R_c \approx \frac{1}{L}h(X) - \frac{1}L\log_2(\nu_c),
\end{equation}
where $h(X)$ is the differential entropy of a source vector and the approximation becomes asymptotically exact in the high resolution limit where $\nu_c \rightarrow 0$. 

The side descriptions are based on a coarser lattice obtained by scaling the Voronoi cells of the central lattice by a factor of $N_i \mu_i$. Assuming the pdf of $X$ is roughly constant within a sublattice cell, the rates of the side descriptions are given by
\begin{equation}\label{eq:Ri}
R_i\approx \frac{1}{L}h(X) - \frac{1}L\log_2(N_i \mu_i \nu_c)
\end{equation}
where the approximation becomes exact asymptotically as $N_i\nu_c=\nu_i\rightarrow 0$ for a fixed weight $0<\mu_i\in \mathbb{R}$. 
The description rates are related to the single-description rate by
\begin{equation*}
R_i \approx R_c - \frac{1}L\log_2(N_i\mu_i).
\end{equation*}

It follows that, given description rates $R_i$ and description weights $\mu_i$ for $i=0,\dotsc,n-1,$ the index values are given by
\begin{equation}
N_i = \frac{1}{\nu_c \mu_i} 2^{h(X)-LR_i}
\end{equation}
and the nesting ratios by $N_i' = N_i^{\frac{1}{L}}$.

\subsection{Distortions}\label{sec:distortions}
We now provide analytical expressions for the expected distortions in the case of $n=2$ and $n=3$ descriptions.

\begin{theorem}\label{theo:dist_2desc}
Let $n=2$ and $1 \leq L\in \mathbb{N}$. Furthermore, fix the weights $0<\mu_i\in\mathbb{R}$ and $\gamma_i\in\mathbb{R}$ where $i=0,1$. Given an optimal labeling function $\alpha$, then, asymptotically as $N_i\rightarrow\infty$ and $\nu_i\rightarrow 0$, the expected distortion $\bar{D}_i = \mathbb{E}\|X - \hat{X}_i\|^2$ where $\hat{X}_i=\mu_i\lambda_i$ satisfies
\begin{align}\label{eq:Di1}
\bar{D}_i &=  \frac{\gamma_j^2}{(\gamma_0 + \gamma_1)^2} G(S_L) \nu_c^{2/L} (N_0N_1)^{2/L}(\mu_0\mu_1)^{2/L} \\ 
&= \frac{\gamma_j^2}{(\gamma_0 + \gamma_1)^2} G(S_L) 2^{\frac{2}{L}h(X)}2^{2(R_c - (R_0+R_1))} \label{eq:Di2}
\end{align}
where $i,j\in\{0,1\}$ and $i\neq j$. 
\end{theorem}

\begin{IEEEproof}
Follows by applying the proof technique of Diggavi et al.~\cite{diggavi:2002} and using the fact that we are here optimizing over $L$-dimensional spheres rather than Voronoi cells as was the case in~\cite{diggavi:2002}.
\end{IEEEproof}

\begin{theorem}\label{theo:dist_3desc}
Let $n=3$ and $1\leq L\in\mathbb{N}$. Given the set of distortion weights $\{\gamma_\ell \in \mathbb{R} : \ell\in \mathcal{L}^{(n,\kappa)}, 1\leq \kappa\leq n-1\}$, and set of description weights $\{0< \mu_{i} \in \mathbb{R} : i=0,\dotsc, n-1\}$ and an optimal labeling function $\alpha$. Then, 
for any $1\leq \kappa <n$, any $\ell \in \mathcal{L}^{(n,\kappa)}$, and asymptotically as $N_i\rightarrow\infty$ and $\nu_i\rightarrow 0$, the expected distortion $\bar{D}_\ell = \mathbb{E}\|X - \hat{X}_\ell\|^2$ where $\hat{X}_\ell = \sum_{i\in \ell}\mu_i\lambda_i$ is given by
\begin{align}\label{eq:Dell}
\bar{D}_\ell &= \hat{\gamma}_{\ell} \Phi_L
G(S_L)  \nu_c^{2/L} (\mu_0\mu_1\mu_2)^{1/L}  (N_0N_1N_2)^{1/L} \\
&= \hat{\gamma}_{\ell} \Phi_L G(S_L) 2^{\frac{2}{L}h(X)}2^{R_c - (R_0+R_1+R_2)} \label{eq:DellRate}
\end{align}
where the weights $\hat{\gamma}_\ell\in\mathbb{R}$ for $\kappa=1$ is given by
\begin{equation}\label{eq:hatgammak1}
\hat{\gamma}_{i} = \frac{\gamma_{j}^2 + \gamma_{k}^2 + \gamma_{j}\gamma_{k}}{(\gamma_{0}+\gamma_{1}+\gamma_{2})^2}
\end{equation}
and for $\kappa=2$ by
\begin{equation}\label{eq:hatgammak2}
\hat{\gamma}_{i,j} = \frac{1}{4}\frac{\gamma_{i,k}^2 + \gamma_{j,k}^2 + \gamma_{i,k}\gamma_{j,k}}{(\gamma_{0,1}+\gamma_{0,2}+\gamma_{1,2})^2}
\end{equation}
where $k\neq i, k\neq j$, and $j\neq i$ and $\Phi_L = \frac{L+2}L \frac{\tilde{\beta}_L}{\beta_L} \psi_{3,L}^{2} $ 
where $\psi_{3,L}$ is given by~(\ref{eq:psi3L}), $\beta_L$ is given by~(\ref{eq:betaL}) and
\begin{equation}\label{eq:tildebetaL}
\begin{split}
\tilde{\beta}_L &=
\sum_{m=0}^{\frac{L+1}2} \binom{\frac{L+1}2}{m}
2^{\frac{L+1}2-m}(-1)^m \sum_{k=0}^{\frac{L-1}2} \frac{
  (\frac{L+1}2)_k (\frac{1-L}2)_k }{ (\frac{L+3}2)_k\, k! } \\
&\times
\sum_{j=0}^{k}\binom{k}{j}\bigg(\frac{1}{2}\bigg)^{k-j}(-1)^j\bigg(\frac{1}4\bigg)^{j}\frac{1}{L+m+j+2}.
\end{split}
\end{equation}
\end{theorem}

\begin{IEEEproof}
See Appendix~\ref{app:theo_3desc}.
\end{IEEEproof}

For large $L$, we can simplify the term $\Phi_L$ appearing in Theorem~\ref{theo:dist_3desc}, which we for later reference put into the following corollary:
\begin{corollary}\label{cor:PhiL}
Asymptotically as $N_i\rightarrow\infty$ and $L\rightarrow \infty$, $\Phi_L = \left(\frac{4}3\right)^{\frac{1}{2}}$.
\hfill$\blacktriangle$
\end{corollary}

If we in the construction of the $n$-tuples use the additional set of weights $\{c_{i,j}\}$ as given by~(\ref{eq:cij}), then $\hat{\gamma}_{\ell}$ is given by the following lemma:
\begin{lemma}\label{lem:c_ij}
For any $n>1, 1\leq \kappa<n$ and $\ell\in \mathcal{L}^{(n,\kappa)}$ we have
\begin{equation}\label{eq:gammacomplex}
\begin{split}
\hat{\gamma}_{\ell} &=
\frac{1}{\bar{\gamma}(\mathcal{L}^{(n,\kappa)})^2\kappa^2}
\bigg(
\bar{\gamma}(\mathcal{L}^{(n,\kappa)})\sum_{j\in \ell}\sum_{\substack{i=0\\ i\neq j}}^{n-1} \bar{\gamma}(\mathcal{L}_i^{(n,\kappa)})c_{i,j} \\
&\quad -\bar{\gamma}(\mathcal{L}^{(n,\kappa)})^2\sum_{i=0}^{\kappa-2}\sum_{j=i+1}^{\kappa-1}
c_{i,j}\\ 
&\quad-\sum_{i=0}^{n-2}\sum_{j=i+1}^{n-1}\bar{\gamma}(\mathcal{L}_i^{(n,\kappa)})\bar{\gamma}(\mathcal{L}_j^{(n,\kappa)})c_{i,j}
\bigg)
\end{split}
\end{equation}
where if $c_{i,j}=1$ and $n=3$,~(\ref{eq:gammacomplex}) reduces to~(\ref{eq:hatgammak1}) and~(\ref{eq:hatgammak2}) for $\kappa=1$ and $\kappa=2$, respectively.
\end{lemma}
\begin{IEEEproof}
Follows by inserting the additional weights $\{c_{i,j}\}$ in~(\ref{eq:sideperm}).
\end{IEEEproof}

Notice also that, for any $n\geq 1$ and asymptotically as $\nu_c\rightarrow 0$, the expected central distortion is trivially given by 
\begin{equation}\label{eq:Dc}
\bar{D}_c =\mathbb{E}D_c = \mathbb{E}\|X - \mathcal{Q}_{\Lambda_c}(X)\|^2 = G(\Lambda_c)\nu_c^{2/L}.
\end{equation}

We end this section by establishing an interesting result for the $n$-channel IA based MD problem.
\begin{corollary}\label{theo:dist_ndesc}
Let $n>1$ and $1\leq L<\infty$. Given the set of distortion weights $\{\gamma_\ell \in \mathbb{R} : \ell\in \mathcal{L}^{(n,\kappa)}, 1\leq \kappa\leq n-1\}$, and set of description weights $\{0< \mu_{i} \in \mathbb{R} : i=0,\dotsc, n-1\}$ and an optimal labeling function $\alpha$. 
Then, for any $1\leq \kappa <n$, any $\ell \in \mathcal{L}^{(n,\kappa)}$, and asymptotically as $N_i\rightarrow\infty$ and $\nu_i\rightarrow 0$, the expected distortion $\bar{D}_\ell = \mathbb{E}\|X - \hat{X}_\ell\|^2$ where $\hat{X}_\ell = \sum_{i\in \ell}\mu_i\lambda_i$ is linearly proportional to $\bar{D}_{\ell'}$ where $\ell'\in \{ \mathcal{L}^{(n,\kappa)} : 1\leq \kappa < n\}$. In particular
\begin{equation}\label{eq:ndist}
\bar{D}_\ell = \hat{\gamma}_{\ell} c_\ell 2^{\frac{2}{L}h(X)} 2^{\frac{2}{n-1}(R_c - \sum_{i=0}^{n-1}R_i)}
\end{equation}
where $\hat{\gamma}_{\ell}$ is given by~(\ref{eq:gammacomplex}) and $0 < c_\ell < \infty, \forall \ell$.
\end{corollary}
\begin{IEEEproof}
Follows from the proof of Theorem~\ref{theo:dist_3desc}.
\end{IEEEproof}

\begin{remark}
We have not been able to find the set of constants $\{c_\ell\}$ in~(\ref{eq:ndist}) for the case of $n>3$. However, since $0<c_\ell < \infty$ it follows that, for any $n>1$, the side distortions for different subsets of descriptions are linearly related, independently of the description rates. This observation has an interesting consequence. 
Let the growth of $N_\pi=\prod_i N_i$ as a function of the rates be given by $N_\pi = 2^{La(n-1)\sum_i R_i }$ where $0<a<1$. 
Moreover, since $R_i = R_c - \frac{1}{L}\log_2( N_i)$ we also have that $N_\pi = 2^{L(nR_c - \sum_i R_i)}$.  Equating the two expressions for $N_\pi$ and solving for $R_c$ yields $R_c = \frac{1}{n}\sum_i R_i ( a(n-1) +1)$. Inserting this into~(\ref{eq:ndist}) and~(\ref{eq:Dc}) leads to
\begin{equation}\label{eq:Dl_sim}
\lim_{\sum_i R_i \to \infty}  \bar{D}_\ell\,  2^{\frac{2}{n}(1-a)\sum_i R_i} = c'\, 2^{\frac{2}{L}h(X)}
\end{equation}
for any $\ell \in \mathcal{L}^{(n,\kappa)}$ and
\begin{equation}\label{eq:Dc_sim}
\lim_{\sum_i R_i \to \infty}  \bar{D}_c \, 2^{\frac{2}{n}(1+a(n-1))\sum_i R_i} = c\, 2^{\frac{2}{L}h(X)},
\end{equation}
where $c'\in\mathbb{R}$ depends on $\ell$, $c\in\mathbb{R}$ is independent of $\ell$ and $a$ controls the rate trade-offs between the central and the side descriptions.
Thus, the product of the central distortion $\bar{D}_c$~(\ref{eq:Dc_sim}) and an arbitrary set of $(n-1)$ side distortions $\bar{D}_\ell$ ~(\ref{eq:Dl_sim}) is independent of $a$. This observation agrees with the symmetric $n$-channel product considered in~\cite{zhang:2010}.
\end{remark}

\section{Comparison to Existing Schemes}\label{sec:comparison}
We first assess the two-channel performance. This is interesting partly because it is the only case where the complete achievable MD rate-distortion region is known and partly because it makes it possible to compare the performance to that of existing schemes. 

\subsection{Two-Channel Performance}
The side distortions $\bar{D}'_0$ and $\bar{D}'_1$ of the two-channel asymmetric MD-LVQ system presented in~\cite{diggavi:2002} satisfy (under identical asymptotical conditions as that of the proposed design)
\begin{equation}\label{eq:digd0}
\bar{D}'_i \approx \frac{\gamma_j^2}{(\gamma_0+\gamma_1)^2}G(\Lambda_\pi)2^{2h(X)}2^{-2(R_0+R_1-R_c)}
\end{equation}
where $i,j\in\{0,1\}$ and $i\neq j$ and the central distortion is given by $\bar{D}'_c\approx G(\Lambda_c)2^{2(h(X)-R_c)}$. 
Notice that the only difference between~(\ref{eq:digd0}) and~(\ref{eq:Di2}) is that the former depends on $G(\Lambda_\pi)$ and the latter on $G(S_L)$. 
For the two dimensional case it is known that $G(S_2)=1/(4\pi)$ whereas if $\Lambda_\pi$ is similar to $Z^2$ we have $G(\Lambda_\pi)=1/12$ which is approximately $0.2$ dB worse than $G(S_2)$. Fig.~\ref{fig:2chan_fix_z2} shows the performance when quantizing $2\cdot 10^6$ 
zero-mean unit-variance independent Gaussian vectors constructed by blocking an i.i.d.\ scalar Gaussian process into two-dimensional vectors and using the $Z^2$ quantizer for the design of~\cite{diggavi:2002} as well as for the proposed system. 
In this setup we have fixed $R_0=5$ bit/dim.\ but $R_1$ is varied in the range $5$ -- $5.45$ bit/dim. We have fixed the ratio $\gamma_0/\gamma_1=1.55$ and we keep the side distortions fixed and change the central distortion. Since the central distortion is the same for the two schemes we have not shown it. Notice that $\bar{D}_0$ (resp.\ $\bar{D}_1$) is strictly smaller (about $0.2$ dB) than $\bar{D}'_0$ (resp.\ $\bar{D}'_1$).
%
%
\begin{figure}[ht]
\psfrag{Theo: d0}{\scriptsize Theo: $\bar{D}_0$}
\psfrag{Theo: d1}{\scriptsize Theo: $\bar{D}_1$}
\psfrag{Num: d0}{\scriptsize Num: $\bar{D}_0$}
\psfrag{Num: d1}{\scriptsize Num: $\bar{D}_1$}
\psfrag{Theo: dddd0}{\scriptsize Theo: $\bar{D}'_0$}
\psfrag{Theo: dd1}{\scriptsize Theo: $\bar{D}'_1$}
\begin{center}
\includegraphics[width=8cm]{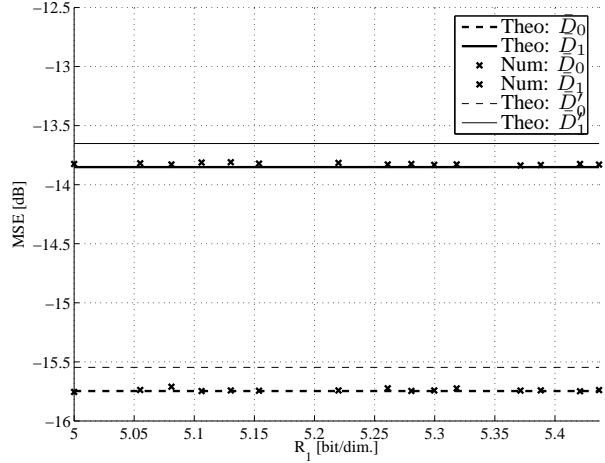}
\caption{The side distortions are here kept fixed as the rate is increased. Notice that the numerically obtained side distortions $\bar{D}_0$ and $\bar{D}_1$ are strictly smaller than the theoretical $\bar{D}'_0$ and $\bar{D}'_1$.}
\label{fig:2chan_fix_z2}
\vspace{-1cm}
\end{center}
\end{figure}

\subsection{Three Channel Performance}
In this section we compare the rate-distortion properties of the
proposed design to the inner bound provided
in~\cite{pradhan:2004,puri:2005}. Thus, we restrict attention to the
symmetric case. 
In order to do this, we first define an MD distortion product for the three channel case. 
Then, we show that by random binning one can further reduce the description rates. 
Finally, we assess the rate loss when finite-dimensional quantizers are used but no binning.

\subsubsection{Three Channel Distortion Product}
To assess the performance of the three channel design it is convenient to define the distortion product $D^{\pi}$ which in the symmetric distortion case (i.e.\ for $\bar{D}_0=\bar{D}_1=\bar{D}_2$ and $\bar{D}_{0,1}=\bar{D}_{0,2}=\bar{D}_{1,2}$) takes the form $D^{\pi}=\bar{D}_c\bar{D}_i\bar{D}_{i,j}$. This is similar in spirit to Vaishampayan's widely used symmetric two-channel distortion product~\cite{vaishampayan:1998}.

Let $n=3$ and consider the symmetric case where $\mu_i=1, \gamma_i=c_1$ and $\gamma_{i,j}=c_2$ for $i,j=0,1,2$ where $c_1,c_2$ are some constants. Moreover, $R_i=R$ and $N_i=N$ for $i=0,1,2$. 
It follows from~(\ref{eq:hatgammak1}) that $\hat{\gamma}_i=\frac{1}3$ and from~(\ref{eq:hatgammak2}) that $\hat{\gamma}_{i,j} = \frac{1}{12}$ so that by~(\ref{eq:DellRate}) we see that the one-channel distortion is given by
\begin{equation}
\bar{D}_i = \frac{1}3 \Phi_L G(S_L) 2^{\frac{2}{L}h(X)+ R_c - 3R}
\end{equation}
and the two-channel distortion is given by
\begin{equation}\label{eq:Dij}
\bar{D}_{i,j} = \frac{1}{12} \Phi_L G(S_L) 2^{\frac{2}{L}h(X)+ R_c - 3R}.
\end{equation}
We also recall that the central distortion is given by
\begin{equation}
\bar{D}_c = G(\Lambda_c)2^{\frac{2}{L}h(X)-2R_c}.
\end{equation}
This leads to the following distortion product
\begin{equation}\label{eq:3chan_distprod_L}
D^\pi = \frac{1}{36} \Phi_L^2 G(S_L)^2G(\Lambda_c) 2^{\frac{6}{L}h(X) - 6R}
\end{equation}
which is independent of $R_c$ and only depends upon the description rate  $R$.

Recall that in the Gaussian case, $h(X) = \frac{L}{2}\log_2(2\pi e \sigma_X^2)$ and for $L\rightarrow\infty$ we have $G(S_L)=G(\Lambda_c)=1/(2\pi e)$ and (by Corollary~\ref{cor:PhiL}) $\Phi_\infty^2 = \frac{4}3$ so that the distortion product reduces to
\begin{equation}\label{eq:3chan_distprod}
D^\pi = \frac{1}{27}\sigma_X^62^{- 6R}.
\end{equation}

The following lemma shows that the proposed design is able to achieve a distortion product based on the inner bound of~\cite{pradhan:2004,puri:2005}.
\begin{lemma}\label{lem:scec}
The high-resolution distortion product $D^\pi$ of the three-channel achievable quadratic Gaussian rate-distortion region of Pradhan et al.~\cite{pradhan:2004,puri:2005} is identical to~(\ref{eq:3chan_distprod}).
\end{lemma}
\begin{IEEEproof}
See Appendix~\ref{app:scec}.
\end{IEEEproof}

\begin{remark}
Thus, for any rate trade-offs between central and side descriptions,
the distortion product of the proposed MDLVQ achieves a distortion
product based on the inner bound
of~\cite{pradhan:2004,puri:2005}. This inner bound is not always tight
as shown in~\cite{tian:2008}. However, in the case where we are only
interested in the one-channel distortion $D_i$ and the central
distortion $D_c$, optimality was recently proven
in~\cite{zhang:2010}. In particular, independently of our work,~\cite{zhang:2010} proposed a distortion product based on the outer bound
of~\cite{venkataramani:2003}. Moreover, it was shown that in
the three-channel case, the product $D_i^2D_c$ of our MDLVQ
construction achieves the distortion product of~\cite{zhang:2010}.
We show next that in the case where we are only interested in the two-channel distortion $D_{i,j}$ and the central distortion $D_c$, we are in fact also optimal. 
\end{remark}

\subsubsection{Random Binning on the Labeling Function}
It was shown in~\cite{pradhan:2004,puri:2005} that the achievable rate region can be enlarged by using random binning arguments on the random codebooks. 
Interestingly, we can show that it also makes sense to apply random binning on the labeling function proposed in this work. 
For example, in the case of three descriptions, we can utilize the universality of random binning so that one can faithfully decode on reception of e.g.\ at least two of the three descriptions. With such a strategy, it is then possible to reduce the effective description rate, since the binning rate is smaller than the codebook rate. The price to pay is that one cannot faithfully decode if e.g.\ only a single description is received.

In order to understand how we apply random binning on the labeling function, recall that every $\lambda_i\in \Lambda_i$ is combined with the set of $\lambda_j$'s given by $\mathcal{T}_j(\lambda_i)\triangleq \{\lambda_j \in \Lambda_j : \lambda_i = \alpha_i(\lambda_c), \lambda_j=\alpha_j(\lambda_c), \lambda_c\in\Lambda_c\}$. The trick is now to randomly assign members of $\mathcal{T}_j(\lambda_i)$ to a set of bins in such a way that it is very unlikely that two or more members of $\mathcal{T}_j(\lambda_i)$ fall into the same bin. When encoding, we first apply the central quantizer $\mathcal{Q}_{\Lambda_c}$ on the source variable $X$ in order to obtain the central lattice point $\lambda_c=\mathcal{Q}_{\Lambda_c}(X)$. We then map the given $\lambda_c$ to the triplet $(\lambda_0,\lambda_1,\lambda_2)=\alpha(\lambda_c)$. We finally find and transmit the bin indices of $\lambda_i, i=0,1,2$, rather than their codebook indices. 
On reception of at least (any) two bin indices, we search through all the elements in the two bins in order to find a pair of sublattice points which are elements of the same $n$-tuple. If the binning rate is large enough, there will (with high probability) be only one such pair of sublattice points for any two bin indices.

\begin{theorem}
Let $n=3$ and let $\alpha$ be an optimal labeling function. Moreover, assume we apply random binning on the labeling function such that one can faithfully (and uniquely) decode on reception of any two descriptions. Then, asymptotically, as $N_i\rightarrow\infty, \nu_i\rightarrow 0,$ and $L\to \infty,$ the binning rate $R_b$ must satisfy
\begin{equation}\label{eq:Rb}
R_b > \frac{1}{2} R + \frac{1}{2}\log_2(\psi_{3,L}\sqrt{N'})
\end{equation}
where $R$ is the description rate.
\end{theorem}
\begin{IEEEproof}
The proof is essentially similar to the technique presented in~\cite{pradhan:2004}.\footnote{The complete proof for the asymmetric case can be found in~\cite{ostergaard:2007a}.}
\end{IEEEproof}

We can further show that the binning rate, as given by~(\ref{eq:Rb}), coincide with that of~\cite{pradhan:2004} for this particular case where we can only decode on reception of at least two out of three descriptions. 
To show this, note that when we get arbitrarily close to the binning rate in~(\ref{eq:Rb}), it follows that 
\begin{equation}\label{eq:Rs_bin}
R = 2R_b - \frac{1}{2}\log_2(\psi_{3,L}^2) - \frac{1}{2}\log_2(N').
\end{equation}
In this case, the two-channel distortion $\bar{D}_{i,j}$, as given by~(\ref{eq:Dij}), can be written as
\begin{align} \notag
\bar{D}_{i,j} &= \frac{1}{12} \Phi_L G(S_L) 2^{\frac{2}{L}h(X)}2^{R_c - 3R} \\ \notag
    &\overset{(a)}{=} \frac{1}{12} \psi_{3,\infty}^2 2^{Rc-3R} \\ \notag
    &\overset{(b)}{=} \frac{N'}{12} \psi_{3,\infty}^2 2^{-2R}  \\ 
    &\overset{(c)}{=} \frac{(N')^2}{12} \psi_{3,\infty}^4 2^{-4R_b}
\end{align}
where $(a)$ is valid for (unit-variance) Gaussian sources, in the limit as $L\rightarrow\infty$ so that $\Phi_\infty = \psi_{3,\infty}^2$ and $2^{\frac{2}{L}h(X)}=G(S_L)^{-1}$. $(b)$ follows since $R_c = R + \log_2(N')$ and $(c)$ follows by inserting~(\ref{eq:Rs_bin}). Similarly, in the limit as $L\rightarrow \infty$, the three-channel distortion (central distortion $D_c$) is given by
\begin{align}
D_c &= 2^{-Rc} \\
    &= \frac{1}{N'} \psi_{3,\infty}^2 2^{-4R_b}. \label{eq:Dcbin}
\end{align}

On the other hand, from~\cite{pradhan:2004}, see also Appendix~\ref{app:scec}, it follows that the two-channel distortion $D'_{i,j}$ of Pradhan et al., is given by
\begin{align}
D'_{i,j} &= \frac{1}{2} \sigma_q^2(1+\rho)
\end{align}
where $\rho$ is defined in Appendix~\ref{app:scec} and
\begin{equation}\label{eq:sigmaq}
\sigma_q^2 = 2(1-\rho)^{-1/3}(1+2\rho)^{-2/3}2^{-4R_b}.
\end{equation}
Moreover, the three-channel distortion $D'_{i,j,k}$ is given by
\begin{equation}\label{eq:D3}
D'_{i,j,k} = \frac{1}{3}\sigma_q^2(1+2\rho).
\end{equation}
Let us equate the pair of two-channel distortions, i.e.\ $\bar{D}_{i,j}=D'_{i,j}$, from which we obtain
\begin{equation}\label{eq:equalizeD2}
(1+2\rho)^{1/3} = \left(\frac{1}{12}\psi_{3,\infty}^4(N')^2\right)^{-1/2} (1+\rho)^{1/2}(1-\rho)^{-1/6}.
\end{equation}
Inserting~(\ref{eq:equalizeD2}) and~(\ref{eq:sigmaq}) into~(\ref{eq:D3}) yields
\begin{align}
D'_{i,j,k} &= \frac{2}{3}\!\! \left(\frac{1}{12}\psi_{3,\infty}^4(N')^2\right)^{-1/2}\!\!\!\! (1+\rho)^{1/2}(1-\rho)^{-1/2}2^{-4R_b}\\
&= \frac{1}{N'}\left(\frac{4}{3}\right)^{\frac{1}{2}}2^{-4R_b} \label{eq:3chandist}
\end{align}
where the last equality follows by inserting $\psi_{3,\infty}^2= \left(\frac{4}{3}\right)^{\frac{1}{2}}$ and letting $\rho\rightarrow -\frac{1}{2}$, which corresponds to the asymptotical case where $N'\rightarrow\infty$. It follows that the resulting two and three-channel distortions are identical (the ratio of~(\ref{eq:Dcbin}) and~(\ref{eq:3chandist}) is one) for the the proposed design and the bounds of Pradhan et al.~\cite{pradhan:2004}.

\subsubsection{Rate Loss}
Let us define a rate loss for the symmetric case as $R_{\text{loss}}  \triangleq R -
R^{\text{inn}}$ (per description), where $R^{\text{inn}}$ denotes the
rate obtained from the inner bound of Pradhan et al. With this, the
rate loss can easily be derived from the distortion product by isolating the rates in (\ref{eq:3chan_distprod_L}) and (\ref{eq:3chan_distprod}) and forming their difference, that is
\begin{equation}\label{eq:rateloss}
\begin{split}
R_{\text{loss}} &= \frac{1}{6}\log_2(\Phi_L^2) +
\frac{1}{6}\log_2\left(\frac{3}{4}\right) \\
&\quad+ \frac{1}{6}\log_2(G(S_L)^2G(\Lambda_c) (2\pi e)^3)
\end{split}
\end{equation}
which clearly goes to zero for large $L$ since
$\Phi_\infty^2=\frac{4}3$. With this definition of rate loss, 
the scalar rate loss (i.e.\ for $L=1$) is $R_{\text{loss}}= 0.2358$ bit/dim.\ whereas for $L=3$ and using the BCC lattice, the rate loss is $0.1681$ bit/dim. Furthermore, we have numerically evaluated the terms $\log_2(G(S_L)2\pi e)$ and $\log_2(\Phi_L^2\frac{3}4)$ for $1\leq L\leq 21$ (and $L$ odd) as shown in Fig.~\ref{fig:rateloss}. It may be noticed that $\log_2(\Phi_L^2\frac{3}4)$ is strictly smaller than $\log_2(G(S_L)2\pi e)$. It follows that, at least for this range of dimensions, the overall description rate loss, as given by~(\ref{eq:rateloss}), is less than the space-filling loss of the lattice in question. This is in contrast to, for example, the MD scheme presented in~\cite{chen:2006} where the description rate loss is larger than the space-filling loss of the lattices being used. At high dimensions, the rate loss vanishes for both schemes.

\begin{figure}[ht]
\psfrag{GGGSL2pie}{\scriptsize $\log_2(G(S_L)2\pi e)$}
\psfrag{Psi34}{\scriptsize $\log_2(\Phi_L^2\frac{3}4)$}
\begin{center}
\includegraphics[width=7cm]{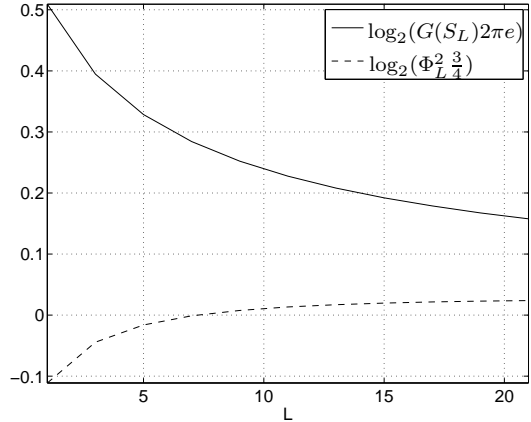}
\caption{The terms $\log_2(G(S_L)2\pi e)$ and $\log_2(\Phi_L^2\frac{3}4)$ as a function of the dimension $L$. Both terms converge to $0$ in the limit as $L\rightarrow\infty$. Notice that $\log_2(G(S_L)2\pi e)>\log_2(\Phi_L^2\frac{3}4)$ in the range shown.}
\label{fig:rateloss}
\end{center}
\end{figure}

\section{Conclusions}\label{sec:conclusion}
We proposed a simple method for constructing IA based $n$-channel asymmetric MD-LVQ schemes. 
For the class of IA based schemes using a single IA function and
averaging reconstruction rules, 
the design was shown to be asymptotical optimal  for any number of descriptions.
For two descriptions, the rate loss was smaller than that of existing
IA based designs whereas for three descriptions, the rate loss (when
compared to the inner-bound of Pradhan et al.\ and restricted to the
case of symmetric rates and distortions) was smaller than that of source splitting. 
It was finally shown that the rate-distortion performance achieves points on the inner bound proposed by Pradhan et al.

\appendices

\section{Proof of Theorem~\ref{theo:Jreduced}}
\label{app:proof_Jreduced}

To prove Theorem~\ref{theo:Jreduced} we need the following results.
\begin{lemma}\label{lem:sumLj}
For $1 \leq \kappa \leq n$ and any $i\in \{0,\dots, n-1\}$ we have
\begin{equation*}
\sum_{\substack{j=0\\j\neq i}}^{n-1}\bar{\gamma}(\mathcal{L}^{(n,\kappa)}_j)=\kappa \bar{\gamma}(\mathcal{L}^{(n,\kappa)})-\bar{\gamma}(\mathcal{L}^{(n,\kappa)}_i).
\end{equation*}
\end{lemma}
\begin{IEEEproof}
Since $|\mathcal{L}^{(n,\kappa)}_j| = \binom{n-1}{\kappa-1}$ the sum $\sum_{j=0}^{n-1}\bar{\gamma}(\mathcal{L}^{(n,\kappa)}_j)$ contains $n\binom{n-1}{\kappa-1}$ terms. However, the number of distinct terms is $|\mathcal{L}^{(n,\kappa)}|=\binom{n}{\kappa}$ and each term is then used $\kappa$ times, since
\begin{equation*}
\frac{n\binom{n-1}{\kappa-1}}{\binom{n}{\kappa}} = \kappa.
\end{equation*}
Subtracting the terms for $j=i$ proves the lemma.
\end{IEEEproof}

\begin{lemma}\label{lem:sumLij}
For $1 \leq \kappa \leq n$ and any $i,j\in \{0,\dots, n-1\}$ we have
\begin{equation*}
\sum_{j=0}^{n-1}\bar{\gamma}(\mathcal{L}^{(n,\kappa)}_{i,j})=\kappa \bar{\gamma}(\mathcal{L}^{(n,\kappa)}_i).
\end{equation*}
\end{lemma}
\begin{IEEEproof}
It is true that $\mathcal{L}^{(n,\kappa)}_{i,i}=\mathcal{L}^{(n,\kappa)}_i$ and since $|\mathcal{L}^{(n,\kappa)}_{i}|=\binom{n-1}{\kappa-1}$ and $|\mathcal{L}^{(n,\kappa)}_{i,j}|=\binom{n-2}{\kappa-2}$ the sum $\sum_{j=0}^{n-1}\bar{\gamma}(\mathcal{L}^{(n,\kappa)}_{i,j})$ contains $(n-1)\binom{n-2}{\kappa-2} + \binom{n-1}{\kappa-1}$ terms. However, the number of distinct $l\in \mathcal{L}^{(n,\kappa)}_i$ terms is $|\mathcal{L}^{(n,\kappa)}_i|=\binom{n-1}{\kappa-1}$ and each term is then used $\kappa$ times, since
\begin{equation*}
\frac{(n-1)\binom{n-2}{\kappa-2} + \binom{n-1}{\kappa-1}}{\binom{n-1}{\kappa-1}} = \kappa.
\end{equation*}
\rspace
\end{IEEEproof}

\begin{lemma}\label{lem:suminnerprod}
For $1 \leq \kappa \leq n$ we have
\begin{equation*}
\sum_{\ell\in\mathcal{L}^{(n,\kappa)}} \gamma_\ell\left\langle\lambda_c,\sum_{i\in \ell} \lambda_i \right\rangle = \left\langle \lambda_c, \sum_{i=0}^{n-1}\lambda_i \bar{\gamma}(\mathcal{L}^{(n,\kappa)}_i)\right\rangle.
\end{equation*}
\end{lemma}
\begin{IEEEproof}
Follows immediately since $\mathcal{L}^{(n,\kappa)}_i$ denotes the set of all $\ell$-terms that contains the index $i$.
\end{IEEEproof}

\begin{lemma}\label{lem:normsum0}
For $1 \leq \kappa \leq n$ we have
\begin{equation*}
\begin{split}
\sum_{i=0}^{n-2}&\sum_{j=i+1}^{n-1}\bar{\gamma}(\mathcal{L}^{(n,\kappa)}_i)\bar{\gamma}(\mathcal{L}^{(n,\kappa)}_j)\|\lambda_i-\lambda_j\|^2
\\
&=
\sum_{i=0}^{n-1}\bar{\gamma}(\mathcal{L}^{(n,\kappa)}_i)\left(\kappa \bar{\gamma}(\mathcal{L}^{(n,\kappa)})-\bar{\gamma}(\mathcal{L}^{(n,\kappa)}_i)\right)\|\lambda_i\|^2 \\
&\quad - 2\sum_{i=0}^{n-2}\sum_{j=i+1}^{n-1}\bar{\gamma}(\mathcal{L}^{(n,\kappa)}_i)\bar{\gamma}(\mathcal{L}^{(n,\kappa)}_j)\langle \lambda_i,\lambda_j\rangle.
\end{split}
\end{equation*}
\end{lemma}
\begin{IEEEproof}
We have that
\begin{equation*}
\begin{split}
\sum_{i=0}^{n-2}&\sum_{j=i+1}^{n-1}\bar{\gamma}(\mathcal{L}^{(n,\kappa)}_i)\bar{\gamma}(\mathcal{L}^{(n,\kappa)}_j)\|\lambda_i-\lambda_j\|^2
\\
&= 
\sum_{i=0}^{n-2}\sum_{j=i+1}^{n-1}\bar{\gamma}(\mathcal{L}^{(n,\kappa)}_i)\bar{\gamma}(\mathcal{L}^{(n,\kappa)}_j)(\|\lambda_i\|^2 + \|\lambda_j\|^2) \\
&\quad - 2\sum_{i=0}^{n-2}\sum_{j=i+1}^{n-1}\bar{\gamma}(\mathcal{L}^{(n,\kappa)}_i)\bar{\gamma}(\mathcal{L}^{(n,\kappa)}_j)\langle \lambda_i,\lambda_j\rangle.
\end{split}
\end{equation*}
Furthermore, it follows that
{\allowdisplaybreaks
\begin{align*}
&\sum_{i=0}^{n-2}\sum_{j=i+1}^{n-1}\bar{\gamma}(\mathcal{L}^{(n,\kappa)}_i)\bar{\gamma}(\mathcal{L}^{(n,\kappa)}_j)(\|\lambda_i\|^2 + \|\lambda_j\|^2) \\ 
&=\sum_{i=0}^{n-2}\bar{\gamma}(\mathcal{L}^{(n,\kappa)}_i)\|\lambda_i\|^2\sum_{j=i+1}^{n-1}\bar{\gamma}(\mathcal{L}^{(n,\kappa)}_j)
\\
&\quad + \sum_{j=1}^{n-1}\bar{\gamma}(\mathcal{L}^{(n,\kappa)}_j)\|\lambda_j\|^2 \sum_{i=0}^{j-1}\bar{\gamma}(\mathcal{L}^{(n,\kappa)}_i) \\
&=
\sum_{i=0}^{n-1}\bar{\gamma}(\mathcal{L}^{(n,\kappa)}_i)\|\lambda_i\|^2\underbrace{\sum_{j=i+1}^{n-1}\bar{\gamma}(\mathcal{L}^{(n,\kappa)}_j)}_{0\
  \text{for}\ i=n-1} \\
&\quad+ \sum_{j=0}^{n-1}\bar{\gamma}(\mathcal{L}^{(n,\kappa)}_j)\|\lambda_j\|^2\underbrace{\sum_{i=0}^{j-1}\bar{\gamma}(\mathcal{L}^{(n,\kappa)}_i)}_{0\ \text{for}\ j=0} \\
&= \sum_{i=0}^{n-1}\bar{\gamma}(\mathcal{L}^{(n,\kappa)}_i)\|\lambda_i\|^2\left( \sum_{j=0}^{i-1}\bar{\gamma}(\mathcal{L}^{(n,\kappa)}_j)+\sum_{j=i+1}^{n-1}\bar{\gamma}(\mathcal{L}^{(n,\kappa)}_j)\right) \\
&= \sum_{i=0}^{n-1}\bar{\gamma}(\mathcal{L}^{(n,\kappa)}_i)\|\lambda_i\|^2\sum_{\substack{j=0\\j\neq i}}^{n-1}\bar{\gamma}(\mathcal{L}^{(n,\kappa)}_j) \\
&= \sum_{i=0}^{n-1}\bar{\gamma}(\mathcal{L}^{(n,\kappa)}_i)\|\lambda_i\|^2\left( \kappa \bar{\gamma}(\mathcal{L}^{(n,\kappa)}) - \bar{\gamma}(\mathcal{L}^{(n,\kappa)}_i)\right),
\end{align*}}
where the last equality follows by use of Lemma~\ref{lem:sumLj}.
\end{IEEEproof}

\begin{lemma}\label{lem:normsum}
For $1 \leq \kappa \leq n$ we have
\begin{equation*}
\begin{split}
\sum_{i=0}^{n-2}\sum_{j=i+1}^{n-1}\bar{\gamma}(\mathcal{L}^{(n,\kappa)}_{i,j})\|\lambda_i-\lambda_j&\|^2 =
(\kappa
-1)\sum_{i=0}^{n-1}\bar{\gamma}(\mathcal{L}^{(n,\kappa)}_i)\|\lambda_i\|^2
\\
&-
2\sum_{i=0}^{n-2}\sum_{j=i+1}^{n-1}\bar{\gamma}(\mathcal{L}^{(n,\kappa)}_{i,j})\langle \lambda_i,\lambda_j\rangle.
\end{split}
\end{equation*}
\end{lemma}
\begin{IEEEproof}
We have that
\begin{equation*}
\begin{split}
\sum_{i=0}^{n-2}\sum_{j=i+1}^{n-1}\!\bar{\gamma}&(\mathcal{L}^{(n,\kappa)}_{i,j})\|\lambda_i-\lambda_j\|^2\! \\
&= 
\sum_{i=0}^{n-2}\sum_{j=i+1}^{n-1}\bar{\gamma}(\mathcal{L}^{(n,\kappa)}_{i,j})(\|\lambda_i\|^2
+ \|\lambda_j\|^2) \\
&\quad- 2\sum_{i=0}^{n-2}\sum_{j=i+1}^{n-1}\bar{\gamma}(\mathcal{L}^{(n,\kappa)}_{i,j})\langle \lambda_i,\lambda_j\rangle.
\end{split}
\end{equation*}
Furthermore, it follows that
{\allowdisplaybreaks
\begin{align*}
&\sum_{i=0}^{n-2}\sum_{j=i+1}^{n-1}\bar{\gamma}(\mathcal{L}^{(n,\kappa)}_{i,j})(\|\lambda_i\|^2
+ \|\lambda_j\|^2) \\
&=
\sum_{i=0}^{n-2}\sum_{j=i+1}^{n-1}\bar{\gamma}(\mathcal{L}^{(n,\kappa)}_{i,j})\|\lambda_i\|^2
+
\sum_{i=0}^{n-2}\sum_{j=i+1}^{n-1}\bar{\gamma}(\mathcal{L}^{(n,\kappa)}_{i,j})\|\lambda_j\|^2 \\
&= 
\sum_{i=0}^{n-2}\|\lambda_i\|^2\sum_{j=i+1}^{n-1}\bar{\gamma}(\mathcal{L}^{(n,\kappa)}_{i,j})
+\sum_{j=1}^{n-1}\sum_{i=0}^{j-1}\bar{\gamma}(\mathcal{L}^{(n,\kappa)}_{i,j})\|\lambda_j\|^2 \\
&= \sum_{i=0}^{n-1}\|\lambda_i\|^2\underbrace{\sum_{j=i+1}^{n-1}\bar{\gamma}(\mathcal{L}^{(n,\kappa)}_{i,j})}_{0\ \text{for}\ i=n-1}
+\sum_{j=0}^{n-1}\|\lambda_j\|^2\underbrace{\sum_{i=0}^{j-1}\bar{\gamma}(\mathcal{L}^{(n,\kappa)}_{i,j})}_{0\ \text{for}\ j=0} \\
&=\sum_{i=0}^{n-1}\|\lambda_i\|^2\left(\sum_{j=0}^{i-1}\bar{\gamma}(\mathcal{L}^{(n,\kappa)}_{i,j})+\sum_{j=i+1}^{n-1}\bar{\gamma}(\mathcal{L}^{(n,\kappa)}_{i,j})\right) \\
&=\sum_{i=0}^{n-1}\|\lambda_i\|^2\left(\sum_{j=0}^{n-1}\bar{\gamma}(\mathcal{L}^{(n,\kappa)}_{i,j})-\bar{\gamma}(\mathcal{L}^{(n,\kappa)}_{i})\right) \\
&\overset{(a)}{=}\sum_{i=0}^{n-1}\|\lambda_i\|^2\left(\kappa \bar{\gamma}(\mathcal{L}^{(n,\kappa)}_{i})-\bar{\gamma}(\mathcal{L}^{(n,\kappa)}_{i})\right) \\
&=(\kappa-1)\sum_{i=0}^{n-1}\|\lambda_i\|^2\bar{\gamma}(\mathcal{L}^{(n,\kappa)}_{i}),
\end{align*}}
where $(a)$ follows by use of Lemma~\ref{lem:sumLij}.
\end{IEEEproof}

\begin{lemma}\label{lem:normsum2}
For $1 \leq \kappa \leq n$ we have
\begin{equation*}
\begin{split}
\sum_{\ell\in\mathcal{L}^{(n,\kappa)}}\gamma_\ell\left\|\sum_{i\in \ell}\lambda_i\right\|^2 &=
\kappa
\sum_{i=0}^{n-1}\bar{\gamma}(\mathcal{L}^{(n,\kappa)}_i)\|\lambda_i\|^2 \\
&\quad- \sum_{i=0}^{n-2}\sum_{j=i+1}^{n-1}\bar{\gamma}(\mathcal{L}^{(n,\kappa)}_{i,j})\|\lambda_i-\lambda_j\|^2.
\end{split}
\end{equation*}
\end{lemma}
\begin{IEEEproof}
The set of all elements $\ell$ of $\mathcal{L}^{(n,\kappa)}$ that contains the index $i$ is denoted by $\mathcal{L}^{(n,\kappa)}_i$. Similar the set of all elements that contains the indices $i$ and $j$ is denoted by $\mathcal{L}^{(n,\kappa)}_{i,j}$. From this we see that
\begin{equation*}
\begin{split}
&\sum_{\ell\in\mathcal{L}^{(n,\kappa)}}\gamma_\ell\left\|\sum_{i\in
    \ell}\lambda_i\right\|^2 \\
&=
\sum_{\ell\in\mathcal{L}^{(n,\kappa)}}\gamma_\ell\left(\sum_{i\in \ell}\|\lambda_i\|^2 + 2\sum_{i=0}^{\kappa-2}\sum_{j=i+1}^{\kappa-1}\langle \lambda_{l_i},\lambda_{l_j}\rangle\right) \\
&= \sum_{i=0}^{n-1}\bar{\gamma}(\mathcal{L}^{(n,\kappa)}_i)\|\lambda_i\|^2 + 2\sum_{i=0}^{n-2}\sum_{j=i+1}^{n-1}\bar{\gamma}(\mathcal{L}^{(n,\kappa)}_{i,j})\langle \lambda_{i},\lambda_{j}\rangle.
\end{split}
\end{equation*}
By use of Lemma~\ref{lem:normsum} it follows that
\begin{equation*}
\begin{split}
&\sum_{\ell\in\mathcal{L}^{(n,\kappa)}}\gamma_\ell\left\|\sum_{i\in \ell}\lambda_i\right\|^2 =
\sum_{i=0}^{n-1}\bar{\gamma}(\mathcal{L}^{(n,\kappa)}_i)\|\lambda_i\|^2\\
&\!+\! (\kappa-1)\!\sum_{i=0}^{n-1}\!\bar{\gamma}(\mathcal{L}^{(n,\kappa)}_i)\|\lambda_i\|^2\! -\! \sum_{i=0}^{n-2}\sum_{j=i+1}^{n-1}\bar{\gamma}(\mathcal{L}^{(n,\kappa)}_{i,j})\| \lambda_{i}-\lambda_{j}\|^2 \\
&=\kappa\sum_{i=0}^{n-1}\bar{\gamma}(\mathcal{L}^{(n,\kappa)}_i)\|\lambda_i\|^2 - \sum_{i=0}^{n-2}\sum_{j=i+1}^{n-1}\bar{\gamma}(\mathcal{L}^{(n,\kappa)}_{i,j})\| \lambda_{i}-\lambda_{j}\|^2
\end{split}
\end{equation*}
\rspace
\end{IEEEproof}

We are now in a position to prove the following result:
\begin{lemma}\label{lem:lcmean}
For $1 \leq \kappa \leq n$ we have
\begin{equation}\label{eq:lcmean0}
\begin{split}
&\sum_{\ell\in\mathcal{L}^{(n,\kappa)}}\gamma_\ell\left\|\lambda_c-\frac{1}{\kappa}\sum_{i\in
    \ell}\lambda_i\right\|^2 \\
&=
\bar{\gamma}(\mathcal{L}^{(n,\kappa)})\left\|\lambda_c-\frac{1}{\kappa
    \bar{\gamma}(\mathcal{L}^{(n,\kappa)})}\sum_{i=0}^{n-1}\bar{\gamma}(\mathcal{L}^{(n,\kappa)}_i)
  \lambda_i\right\|^2 \\ 
&\!+\!\frac{1}{\kappa^2}\!\!\sum_{i=0}^{n-2}\sum_{j=i+1}^{n-1}\!\!\bigg(\frac{\bar{\gamma}(\mathcal{L}^{(n,\kappa)}_i)\bar{\gamma}(\mathcal{L}^{(n,\kappa)}_j)}{\bar{\gamma}(\mathcal{L}^{(n,\kappa)})}-\bar{\gamma}(\mathcal{L}^{(n,\kappa)}_{i,j})\!\!\bigg)\!\|\lambda_i-\lambda_j\|^2.
\end{split}
\end{equation}
\end{lemma}
\begin{IEEEproof}
Expansion of the norm on the left-hand-side in~(\ref{eq:lcmean0}) leads to 
{\allowdisplaybreaks
\begin{align*}
&\sum_{\ell\in\mathcal{L}^{(n,\kappa)}}\gamma_\ell\bigg\|\lambda_c-\frac{1}{\kappa}\sum_{i\in
  \ell}\lambda_i\bigg\|^2 \\
&= \sum_{\ell\in\mathcal{L}^{(n,\kappa)}}\gamma_\ell\left(
\|\lambda_c\|^2 - 2\left\langle\lambda_c,\frac{1}{\kappa}\sum_{i\in \ell}\lambda_i\right\rangle + \frac{1}{\kappa^2}\left\|\sum_{i\in \ell}\lambda_i\right\|^2\right) \\
&\overset{(a)}{=}
\bar{\gamma}(\mathcal{L}^{(n,\kappa)})\|\lambda_c\|^2 - 2\left\langle
  \lambda_c ,
  \frac{1}{\kappa}\sum_{i=0}^{n-1}\bar{\gamma}(\mathcal{L}^{(n,\kappa)}_i)\lambda_i\right\rangle
\\
&\quad
+ \frac{1}{\kappa^2}\sum_{\ell\in \mathcal{L}^{(n,\kappa)}}\gamma_\ell\left\|\sum_{i\in \ell}\lambda_i\right\|^2\\
&= \bar{\gamma}(\mathcal{L}^{(n,\kappa)})\left\| \lambda_c -
  \frac{1}{\kappa
    \bar{\gamma}(\mathcal{L}^{(n,\kappa)})}\sum_{i=0}^{n-1}\bar{\gamma}(\mathcal{L}^{(n,\kappa)}_i)\lambda_i\right\|^2
\\
&\quad - \frac{1}{\kappa^2\bar{\gamma}(\mathcal{L}^{(n,\kappa)})}\left\|\sum_{i=0}^{n-1}\bar{\gamma}(\mathcal{L}^{(n,\kappa)}_i)\lambda_i\right\|^2 \!+\!\!\! \sum_{\ell\in \mathcal{L}^{(n,\kappa)}}\!\!\frac{\gamma_\ell}{\kappa^2}\left\|\sum_{i\in \ell}\lambda_i\right\|^2\\
&=\bar{\gamma}(\mathcal{L}^{(n,\kappa)})\left\| \lambda_c -
  \frac{1}{\kappa
    \bar{\gamma}(\mathcal{L}^{(n,\kappa)})}\sum_{i=0}^{n-1}\bar{\gamma}(\mathcal{L}^{(n,\kappa)}_i)\lambda_i\right\|^2 \\
&+ \!\!\!\sum_{\ell\in \mathcal{L}^{(n,\kappa)}}\!
\frac{\gamma_\ell}{\kappa^2}\left\|\sum_{i\in \ell}\lambda_i\right\|^2
\!\!-\!
\frac{1}{\kappa^2\bar{\gamma}(\mathcal{L}^{(n,\kappa)})}\bigg(
\sum_{i=0}^{n-1}\bar{\gamma}(\mathcal{L}^{(n,\kappa)}_i)^2\|\lambda_i\|^2
\\
&\quad+ 2\sum_{i=0}^{n-2}\sum_{j=i+1}^{n-1}\bar{\gamma}(\mathcal{L}^{(n,\kappa)}_i)\bar{\gamma}(\mathcal{L}^{(n,\kappa)}_j)\langle \lambda_i,\lambda_j\rangle\bigg)\\
&\overset{(b)}{=}\bar{\gamma}(\mathcal{L}^{(n,\kappa)})\left\|
  \lambda_c - \frac{1}{\kappa
    \bar{\gamma}(\mathcal{L}^{(n,\kappa)})}\sum_{i=0}^{n-1}\bar{\gamma}(\mathcal{L}^{(n,\kappa)}_i)\lambda_i\right\|^2
\\
&\!+\! 
\frac{1}{\kappa}\sum_{i=0}^{n-1}\bar{\gamma}(\mathcal{L}^{(n,\kappa)}_i)\|\lambda_i\|^2
\!-\!
\frac{1}{\kappa^2}\sum_{i=0}^{n-2}\sum_{j=i+1}^{n-1}\bar{\gamma}(\mathcal{L}^{(n,\kappa)}_{i,j})\|\lambda_i-\lambda_j\|^2
\\
&-\frac{1}{\kappa^2\bar{\gamma}(\mathcal{L}^{(n,\kappa)})}\sum_{i=0}^{n-1}\bar{\gamma}(\mathcal{L}^{(n,\kappa)}_i)^2\|\lambda_i\|^2 \\
& + 
\sum_{i=0}^{n-2}\sum_{j=i+1}^{n-1}\frac{\bar{\gamma}(\mathcal{L}^{(n,\kappa)}_i)\bar{\gamma}(\mathcal{L}^{(n,\kappa)}_j)}{\kappa^2 \bar{\gamma}(\mathcal{L}^{(n,\kappa)})}
\| \lambda_i-\lambda_j\|^2\\
&
-\frac{1}{\kappa^2 \bar{\gamma}(\mathcal{L}^{(n,\kappa)})}\sum_{i=0}^{n-1}\bar{\gamma}(\mathcal{L}^{(n,\kappa)}_i)(\kappa \bar{\gamma}(\mathcal{L}^{(n,\kappa)})-\bar{\gamma}(\mathcal{L}^{(n,\kappa)}_i))\|\lambda_i\|^2 \\
&=\bar{\gamma}(\mathcal{L}^{(n,\kappa)})\left\| \lambda_c -
  \frac{1}{\kappa
    \bar{\gamma}(\mathcal{L}^{(n,\kappa)})}\sum_{i=0}^{n-1}\bar{\gamma}(\mathcal{L}^{(n,\kappa)}_i)\lambda_i\right\|^2
\\ 
&+\frac{1}{\kappa^2}\sum_{i=0}^{n-2}\sum_{j=i+1}^{n-1}\!\!\!\bigg(\!\frac{\bar{\gamma}(\mathcal{L}^{(n,\kappa)}_i)\bar{\gamma}(\mathcal{L}^{(n,\kappa)}_j)}{\bar{\gamma}(\mathcal{L}^{(n,\kappa)})}- \bar{\gamma}(\mathcal{L}^{(n,\kappa)}_{i,j})\!\!\bigg)\!\|\lambda_i-\lambda_j\|^2,
\end{align*}
where $(a)$ follows by use of Lemma~\ref{lem:suminnerprod} and $(b)$ by use of Lemmas~\ref{lem:normsum0} and~\ref{lem:normsum2}.}
\end{IEEEproof}

\section{Proof of Theorem~\ref{theo:separable}}\label{app:separable}
Without loss of generality, let $\mu_i=1, \forall i$. Furthermore, let 
\begin{equation}\label{eq:f}
f=\sum_{\lambda_c\in V_\pi(0)}\sum_{\kappa=1}^{n-1}\sum_{i=0}^{n-2}\sum_{j=i+1}^{n-1}\hat{\gamma}_{i,j}^{(n,\kappa)}\bigg\|\lambda_i - \lambda_j\bigg\|^2
\end{equation}
and
\begin{equation}\label{eq:g}
g=\sum_{\lambda_c\in V_\pi(0)}\sum_{\kappa=1}^{n-1}\bigg\| \lambda_c - \frac{1}{\kappa \bar{\gamma}(\mathcal{L}^{(n,\kappa)})}\sum_{i=0}^{n-1}\bar{\gamma}(\mathcal{L}^{(n,\kappa)}_i)\lambda_i\bigg\|^2.
\end{equation}
We prove the theorem by constructing a labeling function which lower bounds $f$ independently of $g$. 
We then show that with this choice of $f$ we have $f\rightarrow\infty$ and $g\rightarrow\infty$ but $g/f\rightarrow 0$ as $N_i\rightarrow\infty,\nu_i\rightarrow 0, \forall i$. Furthermore, we show that this holds for any admissible choice of $g$.
Since $\mathcal{J}^n = c_0f + c_1g$ for some constants $c_0,c_1\in \mathbb{R}$ it follows that in order to minimize $\mathcal{J}^n$ an optimal labeling function must jointly minimize $f$ and $g$. 
However, a jointly optimal labeling function can never improve upon the lower bound on $f$ which occur when $f$ is independently minimized. 
Furthermore, $g$ can only be reduced if taking into account during the optimization. Thus, for any optimal labeling function we must have $g/f\rightarrow 0$. It follows that $f$ is asymptotically dominating and therefore must be minimized in order to minimize $\mathcal{J}^n$.

Let $\mathcal{T}$ denote the set of $n$-tuples assigned to central lattice points in $V_\pi(0)$ so that $|\mathcal{T}|=N_\pi$ and let $\mathcal{T}_i$ be the set of $i$th elements (i.e.\ a set of sublattice points all from $\Lambda_i$). Moreover, let $\mathcal{T}(\lambda_i)$ be the set of $n$-tuples containing a specific $\lambda_i$ as the $i$th element. Finally, let $\mathcal{T}_j(\lambda_i)$ be the set of $\lambda_j\in\Lambda_j$ sublattice points which are the $j$th elements in the $n$-tuples that has the specific $\lambda_i$ as the $i$th element.
With this, for any fixed $\lambda_0\in \mathcal{T}_0$, the sum $\sum_{\lambda_1 \in \mathcal{T}_1(\lambda_0)}\|\lambda_0 - \lambda_1\|^2$ runs over the set of $\lambda_1$ points which are in the same $n$-tuples as the given $\lambda_0$. 
Notice that this sum can be written as $\sum_{\lambda_1 \in \mathcal{T}^u_1(\lambda_0)}\#_{\lambda_1}\|\lambda_0 - \lambda_1\|^2$ where the superscript ${}^u$ denotes the unique $\lambda_1$ elements of $\mathcal{T}_1(\lambda_0)$ and $\#_{\lambda_1}$ denotes the number of times the given $\lambda_1$ is used.
Clearly, this sum is minimized if the unique $\lambda_1$ points are as close as possible to the given $\lambda_0$. In other words, for any given ``distribution'' $\{\#_{\lambda_1}\}$, the sum is minimized if the $\lambda_1$'s are contained within the smallest possible sphere around $\lambda_0$. 
In fact, this holds for any $\lambda_0\in \mathcal{T}_0$. On the other hand, keeping the set of $\lambda_1$'s fixed we can also seek the minimizing distribution $\{\#_{\lambda_j}\}$. A good choice appears to be that the $\lambda_j$ points that are closer to the given $\lambda_0$ should be used more frequently than those further way. 

We pause to make the following observation. Due to the shift-invariance property of the labeling function, we can restrict attention to the $n$-tuples which are assigned to central lattice points within $\Lambda_\pi(0)$. Thus, we have a total of $N_\pi$ $n$-tuples. Recall that we guarantee the shift-invariance property by restricting $\lambda_0$ to be inside $V_\pi(0)$ (a restriction which we later relax by considering cosets). 
Furthermore, to avoid possible bias towards any $\lambda_0\in V_\pi(0)$, we require that each $\lambda_0$ is used an equal amount of times. Since there are $N_\pi/N_0$ distinct $\lambda_0$ points in $V_\pi(0)$ it follows that each $\lambda_0$ must be used $N_0$ times.

Let us for the moment being consider the case of $n=3$, i.e.\ we need to construct a set of $N_\pi$ triplets $\mathcal{T}=\{(\lambda_0,\lambda_1,\lambda_2)\}$. If we fix some $\lambda_0$, we can construct a set of pairs of sublattice points by centering a sphere $\tilde{V}$ at $\lambda_0$ and forming the set of distinct pairs $\mathcal{S}=\{(\lambda_0,\lambda_1) : \lambda_1 \in \tilde{V}(\lambda_0) \cap \Lambda_1\}$. For each pair $(\lambda_0',\lambda_1')\in \mathcal{S}$ we can form a triplet $(\lambda_0',\lambda_1',\lambda_2)$ by combining the given pair with some $\lambda_2$. It is important that $\lambda_2$ is close to $\lambda_0'$ as well as $\lambda_1'$ in order to reduce the distances $\|\lambda_0'-\lambda_2\|^2$ and $\|\lambda_1'-\lambda_2\|^2$. 
This can be done by guaranteeing that $\lambda_2 \in \tilde{V}(\lambda_0)$ and $\lambda_2 \in \tilde{V}(\lambda_1)$. In other words, $\lambda_2 \in \tilde{V}(\lambda_0) \cap \tilde{V}(\lambda_1)$. With this strategy, fix some $\lambda_0\in V_\pi(0)$ and start by using some ``small'' $\tilde{V}$ in order to construct the set of pairs $\tilde{S}=\{(\lambda_0,\lambda_1) : \lambda_1 \in \tilde{V}(\lambda_0) \cap \Lambda_1\}$. Then for each pair $s\in \tilde{S}$ we construct the set of triplets $\{ (s,\lambda_2) : \lambda_2 \in \tilde{V}(\lambda_0) \cap \tilde{V}(\lambda_1)\}$. Recall that we need $N_0$ triplets for each $\lambda_0$. However, since $\tilde{V}$ was chosen ``small'' we end up with too few triplets. The trick is now to increase the volume of $\tilde{V}$ in small steps until we end up with exactly $N_0$ triplets (keep in mind that for large $N_0$ we work with large volumes). 

In the $n$-description case, we require that $\lambda_k \in \tilde{V}(\lambda_0) \cap \tilde{V}(\lambda_1) \cap \cdots \cap \tilde{V}(\lambda_{k-1})$. If we let $r$ be the radius of $\tilde{V}$, then with the above procedure it is guaranteed that $\|\lambda_i-\lambda_j\|^2 \leq r^2/L$ for all $(i,j)$ where $i,j=0,\dotsc,n-1$.

Notice that $f$, i.e. the expression to be minimized as given by~(\ref{eq:f}), includes weights $\hat{\gamma}_{i,j}^{(n,\kappa)}$ (which might not be equal) for every pair of sublattices. In otherwords, we might use spheres $\tilde{V}_{i,j}$ of different sizes to guarantee that $\|\lambda_i-\lambda_j\|^2\leq r_{i,j}^2/L$ where the radius $r_{i,j}$ now depends on the particular pair of sublattices under consideration. This is illustrated in Fig.~\ref{fig:3spheres} where $r_{i,j}$ denotes the radius of the sphere $\tilde{V}_{i,j}$. 
Here we center $\tilde{V}_{0,1}$ at some $\lambda_0 \in V_\pi(0)$ as illustrated in Fig.~\ref{fig:3spheres} by the solid circle. Then, for any $n$-tuples having this $\lambda_0$ point as first element, we only include $\lambda_1$ points which are inside $\tilde{V}_{0,1}(\lambda_0)$. This guarantees that $\|\lambda_0 - \lambda_1\|^2 \leq r_{0,1}/L$. 
Let us now center a sphere $\tilde{V}_{1,2}$ at some $\lambda_1$ which is inside $\tilde{V}_{0,1}(\lambda_0)$. This is illustrated by the dotted sphere of radius $r_{1,2}$ in the figure. We then only include $\lambda_2$ points which are in the intersection of $\tilde{V}_{1,2}(\lambda_1)$ and $\tilde{V}_{0,2}(\lambda_0)$. This guarantees that $\|\lambda_i - \lambda_j\|^2 \leq r_{i,j}/L$ for all $(i,j)$ pairs. 
Interestingly, from Fig.~\ref{fig:3spheres} we see that $r_{0,2}$ cannot be greater than $r_{0,1} + r_{1,2}$. Thus, the radius $r_{i,j}$ must grow at the same rate for any pair $(i,j)$ so that, without loss of generality, $r_{0,1}=a_2r_{0,2}=a_1r_{1,2}$ for some fixed $a_1,a_2\in \mathbb{R}$. 
\begin{figure}
\psfrag{V01}{$\tilde{V}_{0,1}$}
\psfrag{V02}{$\tilde{V}_{0,2}$}
\psfrag{V12}{$\tilde{V}_{1,2}$}
\psfrag{r01}{$r_{0,1}$}
\psfrag{r02}{$r_{0,2}$}
\psfrag{r12}{$r_{1,2}$}
\psfrag{l0}{$\lambda_0$}
\psfrag{l1}{$\lambda_1$}
\begin{center}
\includegraphics[width=5cm]{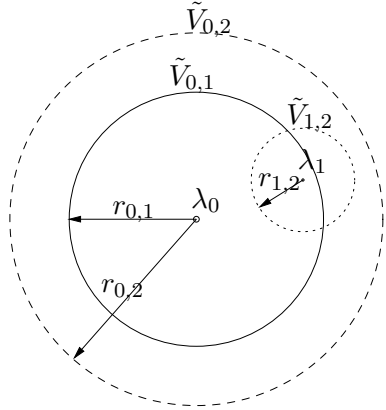}
\caption{Different sizes of the spheres.}
\label{fig:3spheres}
\vspace{-1cm}
\end{center}
\end{figure}


Recall that, the number $\tilde{N}_i$ of lattice points of $\Lambda_i$ within a connected region $\tilde{V}$ of $\mathbb{R}^L$ may be approximated by $\tilde{N}_i\approx \tilde{\nu}/\nu_i$ where $\tilde{\nu}$ is the volume of $\tilde{V}$. 
Moreover, the number of $\lambda_0$ points within $V_\pi(0)$ is given by $\#_{\lambda_0}\approx \text{Vol}(V_\pi(0))/\nu_0=\nu_cN_\pi/\nu_0$. Since we need to construct a total of $N_\pi$ $n$-tuples to label the $N_\pi$ central lattice points, it follows that each $\lambda_0$ is used $N_\pi/\#_{\lambda_0} = N_0$ times. 
Let us now center a sphere $\tilde{V}$ of volume $\tilde{\nu}$ at some $\lambda_0\in \Lambda_0$. The number of $\lambda_i$ points inside this sphere is asymptotically given by $\tilde{\nu}/\nu_i$. Thus, the number of distinct $n$-tuples we can construct by forming all combinations of sublattice points from $\Lambda_i, i=1,\dotsc,n-1$ within $\tilde{V}$ and using $\lambda_0$ as first element is given by $\tilde{\nu}^{n-1}/(\prod_{i=1}^{n-1}\nu_i)$. Recall that we need $N_0$ $n$-tuples for each $\lambda_0$. Thus, we obtain $N_0 = \tilde{\nu}^{n-1}/(\prod_{i=1}^{n-1}\nu_i)$ from which we see that the volume of the sphere $\tilde{V}$ must satisfy
\begin{equation}\label{eq:tildenu_lb}
\tilde{\nu} \geq \nu_c \prod_{i=0}^{n-1}N_i^{\frac{1}{n-1}}.
\end{equation}
We previously argued that we need to make $\tilde{V}$ large enough so as to be able to create exactly $N_0$ $n$-tuples for each $\lambda_0$ which satisfy $\|\lambda_i - \lambda_j\|^2\leq r^2/L$. Having equality in~(\ref{eq:tildenu_lb}) guarantees that $\|\lambda_0 - \lambda_j\|^2\leq r^2/L$ for $j=1,\dotsc,n-1$ but then we must have $\|\lambda_i - \lambda_j\|^2 > r^2/L$ for some $i\neq 0$. However, since we are aiming at lower bounding $f$ we may indeed proceed by assuming that $\|\lambda_i - \lambda_j\|^2 \leq r^2/L$ for all $i$. Furthermore, the different radii $r_{i,j}$ are related through a multiplicative constant which will not affect the rate of growth of the volumes of $\tilde{V}_{i,j}$ as $N_i\rightarrow\infty$. Thus, we proceed by assuming $r_{i,j}=r$ so that $\tilde{V}_{i,j}=\tilde{V}$. 

We are now in a position to evaluate the following sum
\begin{align}
\sum_{\lambda_j \in \mathcal{T}_j(\lambda_i)}\|\lambda_i - \lambda_j\|^2
&\overset{(a)}{=} 
\sum_{\lambda_j \in \mathcal{T}_j(0)}\|\lambda_j\|^2   \\ \label{eq:sumdist1}
&\overset{}{=} 
\sum_{\lambda_j \in \mathcal{T}^u_j(0)}  \#\lambda_{j}\|\lambda_j\|^2.
\end{align}
The volume of the sphere $\tilde{V}$ is independent of which sublattice point it is centered at, so we may take $\lambda_i=0$ from which $(a)$ follows. Notice that for a fixed $\lambda_0$ and for $n>2$, the set $\mathcal{T}_j(\lambda_i)$ contains several identical $\lambda_j$ elements. 
We therefore use the notation $\mathcal{T}^u_j(\lambda_i)$ to indicate the unique set of $\lambda_j$ elements. 
Furthermore, we use the notation $\#_{\lambda_j}$ to indicate the number of times the given $\lambda_j$ is used. 
Since $\sum_{\lambda_j\in T_j^u(0)} \#_{\lambda_j} = N_0$ it follows that $\sum_{\lambda_j\in T_j^u(0)} \min_j\{\#_{\lambda_j}\} \leq N_0$ so that 
$\min_j\{\#_{\lambda_j}\} \leq N_0/\sum_{\lambda_j\in T_j^u(0)}$. Moreover, $|\lambda_j\in T_j^u(0)| = \tilde{\nu}/\nu_j$ which implies that $\sum_{\lambda_j\in T_j^u(0)} = \tilde{\nu}/\nu_j$ and we therefore have that $\min_j\{\#_{\lambda_j}\} \leq N_0 \nu_j/\tilde{\nu}$. By similar reasoning it is easy to show that $\max_j\{\#_{\lambda_j}\} \geq N_0 \nu_j/\tilde{\nu}$.

We have previously shown that the intersection of any number of (large) spheres of equal radii $r$ which are distanced no further apart than $r$, is positively bounded away from zero~\cite{ostergaard:2004b}. In fact, the volume of the smallest intersection can be lower bounded by the volume of a regular $L$-simplex having side lengths $r$~\cite{ostergaard:2004b}. Recall that the volume $\text{Vol}(\mathfrak{S})$ of a regular $L$-simplex $\mathfrak{S}$ with side length $r$ is given by~\cite{buchholz:1992}
\begin{equation}
\text{Vol}(\mathfrak{S}) = \frac{r^L}{L!}\sqrt{\frac{L+1}{2^L}} = c_Lr^L
\end{equation}
where $c_L$ is constant for fixed $L$. 
It follows that, in the three channel case, $\#_{\lambda_j}$ is lower bounded by $c_L r^L/\nu_k$ where $\nu_k$ is the volume of sublattice with the largest index value. Moreover, for $n\geq 3$ we have that $\#_{\lambda_j}$ is lower bounded by $(c_L r^L/\nu_k)^{n-2}$. 

Interestingly, $\#_{\lambda_j}$ is obviously upper bounded by $(\omega_L r^L/\nu_{k'})^{n-2}$, i.e.\ ratio of the volume of an $L$-dimensional sphere of radius $r$ and the volume of a Voronoi cell of $\Lambda_{k'}$, where $k'$ denotes the sublattice with the smallest index value. Notice that the lower bound is proportional to the upper bound and we have the following sandwhich
\begin{equation}\label{eq:sandwhich}
\begin{split}
\left(\frac{\omega_L r^L}{\nu_{k'}} \right)^{n-2}
&\geq\max_j\{\#_{\lambda_j}\} \geq N_0\nu_j/\tilde{\nu} \\
&
\geq \min_j\{\#_{\lambda_j}\} \geq \left(\frac{c_L r^L}{\nu_{k}} \right)^{n-2}
\end{split}
\end{equation}
where the left and right hand sides of~(\ref{eq:sandwhich}) differ by a constant for any $n$ which implies that there exists a positive constant $c>0$ such that $\min_j\{\#_{\lambda_j}\} \geq c N_0\nu_j/\tilde{\nu}$.

Using the above in~(\ref{eq:sumdist1}) leads to 
{\allowdisplaybreaks
\begin{align*}
\sum_{\lambda_j \in \mathcal{T}_j(\lambda_i)}\|\lambda_i - \lambda_j\|^2 \nu_j
&\overset{}{>} 
\frac{cN_i\nu_j}{\tilde{\nu}}\sum_{\lambda_j \in \mathcal{T}^u_j(0)}\|\lambda_j\|^2 \nu_j  \\
&\overset{(a)}{\approx} \frac{cN_i\nu_j}{\tilde{\nu}}\int_{x\in \tilde{V}}\|x\|^2 dx  \\
&= \frac{cN_i\nu_j}{\tilde{\nu}} LG(S_L)\tilde{\nu}^{1+2/L} \\
&= cN_i\nu_j LG(S_L)\nu_c^{2/L}\prod_{i=0}^{n-1}N_i^{\frac{2}{(n-1)L}}
\end{align*}}
where $G(S_L)$ is the dimensionless normalized second moment of an $L$-dimensional hypersphere 
and $(a)$ follows by replacing the sum by an integral (standard Riemann approximation). This approximations becomes exact asymptotically as $N_i\rightarrow\infty$ and $\nu_i\rightarrow 0$.

We finally see that 
{\allowdisplaybreaks
\begin{align*}
\frac{1}{L}\sum_{\lambda_c\in V_\pi(0)}\bigg\|\lambda_i - \lambda_j\bigg\|^2
&\overset{}{=} \frac{1}{L}\sum_{\lambda_i \in \mathcal{T}_i}\sum_{\lambda_j \in \mathcal{T}_j(\lambda_i)}\|\lambda_i - \lambda_j\|^2 \\
&= \frac{1}{L}\frac{N_\pi}{N_i}\sum_{\lambda_j \in \mathcal{T}_j(\lambda_i)}\|\lambda_i - \lambda_j\|^2 \\
&> cG(S_L)N_\pi \nu_c^{2/L}\prod_{i=0}^{n-1}N_i^{\frac{2}{(n-1)L}}
\end{align*}}
so that
\begin{equation}\label{eq:growthoff}
f=\Omega\left(N_\pi\nu_c^{2/L}\prod_{i=0}^{n-1}N_i^{\frac{2}{(n-1)L}}\right).
\end{equation}

We now upper bound $g$. Notice that $g$ describes the sum of distances between the central lattice points and the weighted average of their associated $n$-tuples. 
By construction, these weighted averages will be distributed evenly through-out $V_\pi(0)$. Thus, the distance of a central lattice point and the weighted average of its associated $n$-tuple can be upper bounded by the covering radius of the sublattice with the largest index value, say $N_k$. This is a conservative upper bound but will do for the proof.\footnote{The worst case situation occur if the weighted centroids are distributed such that the minimal distance between any two centroids is maximized. Notice that the weighted centroids form convex combinations of the sublattice points. Since the weights are less than one, the worst case situation occurs if the weighted centroids are distributed on a lattice with an index value equal to the sublattice with the maximum index value (and therefore also the maximum covering radius). Thus, the bound is indeed valid for an arbitrary set of $n$-tuples and not tied to the specific construction of $n$-tuples used so far in the proof.} The rate of growth of the covering radius of the $k$th sublattice is proportional to $\nu_k^{1/L} = (N_k\nu_c)^{1/L}$. Thus
\begin{equation}
g=\mathcal{O}\left( N_\pi \nu_c^{2/L}N_k^{2/L}\right).
\end{equation}
It follows that 
\begin{equation}
\frac{g}{f} = \Theta\left( \frac{N_k^{2/L}}{\prod_{i=0}^{n-1}N_i^{\frac{2}{(n-1)L}}} \right)
= \Theta\left( N_k^{-\frac{2}{L(n-1)}} \right)
\end{equation}
where the last equality follows since the index values are growing at the same rate so that $N_i=N_k/b_i$ for some constant $b_i\in \mathbb{R}$. The theorem is proved by noting that $N_k^{-\frac{2}{L(n-1)}}\rightarrow 0$ as $N_k\rightarrow \infty$. \hfill\IEEEQEDclosed

\section{Proof of Theorem~\ref{theo:nobias} }\label{app:proof_nobias}
We restrict attention to the case where $V_\pi(0)$ is the Voronoi cell of a product lattice generated by the approach outlined in Section~\ref{sec:existence}.
In this case, the shape of $V_\pi(0)$ can be either hyper cubic, or as the dimension increases, the shape can become more and more spherical. 

\begin{figure}
\psfrag{Vpi}{$V_\pi(0)$}
\psfrag{V}{$\tilde{V}$}
\psfrag{s}{$s$}
\psfrag{s'}{$s'$}
\psfrag{A}{$\mathcal{A}$}
\psfrag{B}{$\mathcal{B}$}
\begin{center}
\includegraphics[width=4cm]{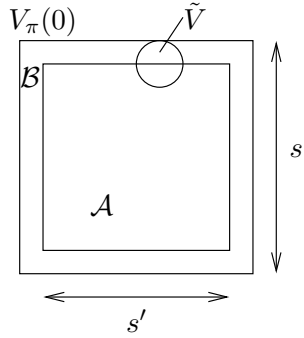}
\caption{$V_\pi(0)$ is a Voronoi cell of $\Lambda_\pi$. $\mathcal{A}$ is a scaled and centered version of $V_\pi(0)$ and $\mathcal{B}$ is the ``strip'' surrounding $\mathcal{A}$, i.e., $\mathcal{B}=V_\pi(0)\backslash\mathcal{A}$. }
\label{fig:Vpi}
\vspace{-1cm}
\end{center}
\end{figure}
Let us first assume that that $V_\pi(0)$ forms a hyper cube having side lengths $s$ as shown in Fig.~\ref{fig:Vpi}. 
The $n$-tuples are constructed by centering a sphere $\tilde{V}$ of volume $\tilde{\nu}$ around each $\lambda_0 \in V_\pi(0)$ and taking all combinations of lattice points within this region (keeping $\lambda_0$ as first coordinate). From Fig.~\ref{fig:Vpi} it may be seen that any $\lambda_0$ which is contained in the region denoted $\mathcal{A}$ will always be combined with sublattice points that are also contained in $V_\pi(0)$. 
On the other hand, any $\lambda_0$ which is contained in region $\mathcal{B}$ will occasionally be combined with points outside $V_\pi(0)$. Therefore, we need to show that the volume $V_\mathcal{A}$ of $\mathcal{A}$ approaches the volume of $V_\pi(0)$ as $N_\pi\rightarrow \infty$ or similarly that the ratio of $V_\mathcal{B}/V_\mathcal{A}\rightarrow 0$ as $N_\pi\rightarrow \infty$, where $V_\mathcal{B}$ denotes the volume of the region $\mathcal{B}$.

Let $\mathcal{A}$ be the centered hyper cube having side lengths  $s' = s - 2\tilde{r}$ where $\tilde{r}$ is the radius of $\tilde{V}$, see Fig.~\ref{fig:Vpi}. 
Since the volume of $V_\pi(0)$ is $\nu_\pi=\nu N_\pi$ it follows that $s = \nu_\pi^{1/L}=(\nu N_\pi)^{1/L}$. Moreover, the volume $V_\mathcal{A}$ of $\mathcal{A}$ is 
{\allowdisplaybreaks
\begin{align*}
V_\mathcal{A} &= (s')^L \\
&= \bigg( s-2\tilde{r} \bigg)^L \\
&= \bigg( \nu_\pi^{1/L}-2 \left(\frac{\tilde{\nu}}{\omega_L}\right)^{1/L} \bigg)^L \\
&= \bigg( (\nu N_\pi)^{1/L}-2 \left(\frac{\psi_{n,L}^L\nu N_\pi^{1/(n-1)}}{\omega_L}\right)^{1/L} \bigg)^L \\
&= \nu\bigg( N_\pi^{1/L}-2 \left(\frac{\psi_{n,L}^L}{\omega_L} \right)^{1/L} N_\pi^{1/L(n-1)} \bigg)^L,
\end{align*}
}%
where $\tilde{\nu} = \psi_{L,n}^L\nu N_\pi^{1/(n-1)}$ is the volume of $\tilde{V}$ and $\tilde{r}=(\tilde{\nu}/\omega_L)^{1/L}$, 
where $\omega_L$ is the volume of an $L$-dimensional unit sphere. 
Since the volume $V_\mathcal{B}$ of $\mathcal{B}$ is given by $V_\mathcal{B}= \nu_\pi -  V_\mathcal{A}$ we find the ratio 
\begin{align*}
&\lim_{N_\pi\to\infty}\frac{ V_\mathcal{B}}{V_\mathcal{A}} \\
&= \lim_{N_\pi\to\infty}\frac{N_\pi}{ \bigg( N_\pi^{1/L}-2 \left(\frac{\psi_{n,L}^L}{\omega_L} \right)^{1/L} N_\pi^{1/L(n-1)} \bigg)^L} -1 \\
&= 0,
\end{align*}
where the second equality follows since $n>2$. 

At this point, we note that the hyper cubic region as used above is actually the worst shape to consider. Specifically, it is the one that yields the minimum $V_\mathcal{A}$ and thus the maximum $V_\mathcal{B}$, since $\nu_\pi$ is constant. To see this, note that we can always pick the region $\mathcal{A}$ to be a centered scaled version of $V_\pi(0)$. Thus, since the boundary of the inscribed region $\mathcal{A}$ will be farthest away from the boundary of $V_\pi(0)$ at corner points it follows that the more spherical $V_\pi(0)$, the larger $V_\mathcal{A}$ compared to $\nu_\pi$. This proves the claim.\hfill\IEEEQEDclosed

\section{Proof of Theorem~\ref{theo:equivtuples}}\label{app:proof_equivtuples}
We only prove it for $\Lambda_0$ and $\Lambda_1$. Then by symmetry it must hold for any pair.
Define the set $\mathcal{S}_{\lambda_0}$ 
as the set of $n$-tuples constructed by centering $\tilde{V}$ at some $\lambda_0\in V_\pi(0)\cap\Lambda_0$.
Hence, $s\in \mathcal{S}_{\lambda_0}$ has $\lambda_0$ as first coordinate and the distance between any two elements of $s$ is less than $r$, the radius of $\tilde{V}$. We will assume\footnote{This is always the case if $r \geq \max_i r(\Lambda_i)$ where $r(\Lambda_i)$ is the covering radius of the $i$th sublattice. The covering radius depends on the lattice and is maximized if $\Lambda_i$ is geometrically similar to $Z^L$, in which case we have\cite{conway:1999}
\begin{equation*}
r(\Lambda_i)=\frac{1}2\sqrt{2}\nu^{1/L}N_i^{1/L}.
\end{equation*}
Since $r=\psi_{n,L} \nu^{1/L} N_\pi^{1/L(n-1)}/\omega_L^{1/L}$ it follows that in order to make sure that $\mathcal{S}_{\lambda_0}\neq\emptyset$ the index values must satisfy
\begin{equation*}\label{eq:Niupperbound}\tag{*}
N_i\leq (\sqrt{2}\psi_{n,L})^L\omega_L N_\pi^{1/(n-1)},\quad i=0,\dots, n-1.
\end{equation*}
Through-out this work we therefore require (and implicitly assume) that~(\ref{eq:Niupperbound}) is satisfied.%
}
that $S_{\lambda_0}\neq \emptyset, \forall \lambda_0$. 

Similarly, define the set $\mathcal{S}_{\lambda_1}\neq \emptyset$ by centering $\tilde{V}$ at some $\lambda_1\in V_\pi(0)\cap\Lambda_1$. 
Recall from Theorem~\ref{theo:nobias} that, asymptotically as $N_i\rightarrow\infty, \forall i$, all elements of the $n$-tuples are in $V_\pi(0)$. 
Then it must hold that for any $s\in \mathcal{S}_{\lambda_1}$ we have $s\in \bigcup_{\lambda_0\in V_\pi(0)\cap\Lambda_0}\mathcal{S}_{\lambda_0}$. But it is also true that for any $s'\in \mathcal{S}_{\lambda_0}$ we have $s'\in \bigcup_{\lambda_1\in V_\pi(0)\cap\Lambda_1}\mathcal{S}_{\lambda_1}$. 
Hence, since the $n$-tuples in $\bigcup_{\lambda_0\in V_\pi(0)\cap\Lambda_0}\mathcal{S}_{\lambda_0}$ are distinct and the $n$-tuples in $\bigcup_{\lambda_1\in V_\pi(0)\cap\Lambda_1}\mathcal{S}_{\lambda_1}$ are also distinct, it follows that the two sets 
$\bigcup_{\lambda_0\in V_\pi(0)\cap\Lambda_0}\mathcal{S}_{\lambda_0}$ and 
$\bigcup_{\lambda_1\in V_\pi(0)\cap\Lambda_1}\mathcal{S}_{\lambda_1}$ must be equivalent.
\hfill\IEEEQEDclosed

\section{Proof of Theorem~\ref{theo:dist_3desc}}
\label{app:theo_3desc}
We notice from Lemma~\ref{lem:expdist_split} and~(\ref{eq:shiftinvJn}) that $\bar{D}_\ell = \mathbb{E}\|X-\hat{X}_\ell\|^2$ can be written as
\begin{equation}\label{eq:barDL_1}
\begin{split}
\bar{D}_\ell &= \frac{1}{N_\pi}\sum_{\lambda_c\in
  V_\pi(0)}\left\|\lambda_c -
  \frac{1}{\kappa}\sum_{i\in\ell}\mu_i\alpha_i(\lambda_c)\right\|^2 \\
&\quad + 
\left\{\sum_{\lambda_c\in \Lambda_c}\int_{V_c(\lambda_c)}f_X(x)\|X-\lambda_c\|^2\, dx\right\}
\end{split}
\end{equation}
where from~(\ref{eq:Dc}) we know that the last term is $G(\Lambda_c)\nu_c^{2/L}$. In the following we therefore focus on finding a closed-form solution to the first term in~(\ref{eq:barDL_1}). This we do by taking the following three steps (which are valid in the usual asymptotical sense):
\begin{enumerate}
\item  We first show, by Proposition~\ref{prop:lclmean}, that 
\begin{equation*}
\begin{split}
&\sum_{\lambda_c\in V_\pi(0)}\left\| \lambda_c -
  \frac{1}{\kappa}\sum_{j\in \ell}\tilde{\lambda}_j\right\|^2 \\
&=
\!\!\!\! \sum_{\lambda_c\in V_\pi(0)}\!\left\| \frac{1}{\kappa}\sum_{j\in \ell}\tilde{\lambda}_j - \frac{1}{\bar{\gamma}(\mathcal{L}^{(n,\kappa)})\kappa}\sum_{i=0}^{n-1} \bar{\gamma}(\mathcal{L}_i^{(n,\kappa)})\tilde{\lambda}_i\right\|^2.
\end{split}
\end{equation*}
\item Then, by Lemma~\ref{lem:squard_sum_equiv}, we show that 
\begin{equation*}
\begin{split}
&\sum_{\lambda_c\in V_\pi(0)}\left\| \frac{1}{\kappa}\sum_{j\in
    \ell}\tilde{\lambda}_j -
  \frac{1}{\bar{\gamma}(\mathcal{L}^{(n,\kappa)})\kappa}\sum_{i=0}^{n-1}
  \bar{\gamma}(\mathcal{L}_i^{(n,\kappa)})\tilde{\lambda}_i\right\|^2\\
&=  \sum_{\lambda_c\in V_\pi(0)}\sum_{k} \sum_{i=0}^{n_k-2}\sum_{i=0}^{n_k-1}c_k \| \tilde{\lambda}_i - \tilde{\lambda_j} \|^2
\end{split}
\end{equation*}
for some $c_k\in\mathbb{R}$ and $n_k\leq n$.
\item Finally, we show by Proposition~\ref{prop:squared_distances} that for the case of $n=3$, we have
\begin{equation}
\sum_{\lambda_c\in V_\pi(0)} \| \tilde{\lambda}_i - \tilde{\lambda_j} \|^2 = 
c \nu_c^{2/L}N_\pi\prod_{m=0}^{2}N_m^{1/L}
\end{equation}
for some $c\in\mathbb{R}$.
\end{enumerate}

In order to establish step 1, we need the following results.

\begin{lemma}\label{lem:normdiff}
For any $1\leq \kappa <n$ and $\ell\in \mathcal{L}^{(n,\kappa)}$ we have
\begin{equation*}
\begin{split}
&\sum_{\lambda_c\in V_\pi(0)}\left\|\frac{1}{\kappa}\sum_{j\in
    \ell}\tilde{\lambda}_j -
  \frac{1}{\bar{\gamma}(\mathcal{L}^{(n,\kappa)})\kappa}\sum_{i=0}^{n-1}\bar{\gamma}(\mathcal{L}_i^{(n,\kappa)})\tilde{\lambda}_i
\right\| \\
&\qquad=
\mathcal{O}\left(\nu_c^{1/L}N_\pi\prod_{m=0}^{n-1}N_m^{1/L(n-1)} \right).
\end{split}
\end{equation*}
\end{lemma}

\begin{IEEEproof}
Recall that the sublattice points $\lambda_i$ and $\lambda_j$ satisfy $\|\lambda_i-\lambda_j\|\leq r/\sqrt{L}$, where $r=(\tilde{\nu}/\omega_L)^{1/L}$ is the radius of $\tilde{V}$. Without loss of generality, we let $\tilde{\lambda}_j=r$ and $\tilde{\lambda}_i=0$, which leads to
\begin{equation*}
\begin{split}
&\sum_{\lambda_c\in V_\pi(0)}\left\|\frac{1}{\kappa}\sum_{j\in
    \ell}\tilde{\lambda}_j -
  \frac{1}{\bar{\gamma}(\mathcal{L}^{(n,\kappa)})\kappa}\sum_{i=0}^{n-1}
  \bar{\gamma}(\mathcal{L}_i^{(n,\kappa)})\tilde{\lambda}_i\right\| \\
&\qquad \leq c\frac{rN_\pi}{\sqrt{L}} 
= \mathcal{O}\left(\nu_c^{1/L}N_\pi\prod_{m=0}^{n-1}N_m^{1/L(n-1)}\right)
\end{split}
\end{equation*}
where $0<c\in \mathbb{R}$  and $\tilde{\nu}=\psi_{n,L}^L\nu_c\prod_{m=0}^{n-1}N_m^{1/(n-1)}$.
\end{IEEEproof}

\begin{proposition}\label{prop:lclmean}
For $1\leq \kappa < n$, $\ell\in \mathcal{L}^{(n,\kappa)}$, $N_i\rightarrow \infty$ and $\nu_i\rightarrow 0$ we have
\begin{equation*}
\begin{split}
&\sum_{\lambda_c\in V_\pi(0)}\left\| \frac{1}{\kappa}\sum_{j\in
    \ell}\tilde{\lambda}_j - \lambda_c\right\|^2\\
& \quad= \sum_{\lambda_c\in V_\pi(0)}\left\| \frac{1}{\kappa}\sum_{j\in \ell}\tilde{\lambda}_j - \frac{1}{\bar{\gamma}(\mathcal{L}^{(n,\kappa)})\kappa}\sum_{i=0}^{n-1} \bar{\gamma}(\mathcal{L}_i^{(n,\kappa)})\tilde{\lambda}_i\right\|^2
\end{split}
\end{equation*}
where $\tilde{\lambda}_j=\mu_j\alpha_j(\lambda_c)$.
\end{proposition}
\begin{IEEEproof}
Let $\bar{\lambda} = \frac{1}{\bar{\gamma}(\mathcal{L}^{(n,\kappa)})\kappa}\sum_{i=0}^{n-1}\bar{\gamma}(\mathcal{L}_i^{(n,\kappa)})\tilde{\lambda}_i$ and $\lambda' = \frac{1}{\kappa}\sum_{j\in \ell}\tilde{\lambda}_j$. 
After some manipulations similar to~\cite[Eqs. (67) -- (72)]{diggavi:2002} we obtain the following inequalities:
\begin{align}\notag%
&\bigg( 1 - 2\frac{\sum_{\lambda_c\in V_\pi(0)}\|\lambda' - \bar{\lambda}\|\|\bar{\lambda}-\lambda_c\|}
{\sum_{\lambda_c\in V_\pi(0)}\|\lambda' - \bar{\lambda}\|^2}\bigg)\!\!\!
\sum_{\lambda_c\in V_\pi(0)} \|\lambda' -
\bar{\lambda}\|^2
\\ \label{eq:cauchy1}
&\leq \sum_{\lambda_c\in V_\pi(0)}\|\lambda_c - \lambda'\|^2 \\ \notag
\leq
&\bigg( \sum_{\lambda_c\in V_\pi(0)} \|\lambda' -
\bar{\lambda}\|^2\bigg)\times\bigg( 1 + \frac{\sum_{\lambda_c\in V_\pi(0)}\|\bar{\lambda}-\lambda_c\|^2}
{\sum_{\lambda_c\in V_\pi(0)}\|\lambda' - \bar{\lambda}\|^2} \\ \label{eq:cauchy2}.
&\quad
+2\frac{\sum_{\lambda_c\in V_\pi(0)}\|\lambda' - \bar{\lambda}\|\|\bar{\lambda}-\lambda_c\|}
{\sum_{\lambda_c\in V_\pi(0)}\|\lambda' - \bar{\lambda}\|^2}\bigg) 
\end{align}
We now use the fact that $\|\bar{\lambda}-\lambda_c\| = \mathcal{O}(N_k\nu_c)^{1/L}$, i.e.\ we can upper this distance by the covering radius of the sublattice with the largest index value, say $N_k$.
By use of Lemma~\ref{lem:normdiff}, it is possible to upper bound the
numerator of the fraction in the l.h.s.\ of~(\ref{eq:cauchy1}) by
\begin{equation}\label{eq:growth_numerator}
\begin{split}
&\sum_{\lambda_c\in V_\pi(0)}\|\lambda' - \bar{\lambda}\|\|\bar{\lambda}-\lambda_c\| \\
&= \mathcal{O}\left((N_k\nu_c)^{1/L} N_\pi\nu_c^{1/L}\prod_{m=0}^{n-1}N_m^{1/L(n-1)}\right).
\end{split}
\end{equation}
At this point we recall that the growth of the denominator in the
l.h.s.\ of~(\ref{eq:cauchy1}) is at least as great as~(\ref{eq:growthoff}), which leads to the following lower bound
\begin{equation}\label{eq:growth_denominator}
\sum_{\lambda_c \in V_\pi(0)}\|\lambda' - \bar{\lambda}\|^2 = \Omega\left(\nu_c^{2/L}N_\pi\prod_{m=0}^{n-1}N_m^{2/L(n-1)}\right).
\end{equation}
By comparing~(\ref{eq:growth_numerator}) to~(\ref{eq:growth_denominator}), it follows that the fractions in~(\ref{eq:cauchy1}) and~(\ref{eq:cauchy2}) go to zero asymptotically as $N_i\rightarrow \infty$. It follows that 
\begin{equation}
\sum_{\lambda_c\in V_\pi(0)}\|\lambda_c - \lambda'\|^2 = \sum_{\lambda_c\in V_\pi(0)}\|\lambda' - \bar{\lambda}\|^2
\end{equation}
which completes the proof.
\end{IEEEproof}

In order to establish step 2, we need the following results.

\begin{lemma}\label{lem:sumjside}
For $1\leq \kappa < n$ and any $\ell \in \mathcal{L}^{(n,\kappa)}$ we have
\begin{equation*}
\left\|\sum_{j\in \ell} \tilde{\lambda}_j\right \|^2=
\kappa\sum_{j\in \ell}\|\tilde{\lambda}_j\|^2 - \sum_{i=0}^{\kappa-2}\sum_{j=i+1}^{\kappa-1}\|\tilde{\lambda}_{l_j}-\tilde{\lambda}_{l_i}\|^2.
\end{equation*}
\end{lemma}
\begin{IEEEproof}
We expand the norm as follows
\begin{align*}
\left\|\sum_{j\in \ell} \tilde{\lambda}_j\right \|^2 &=
\sum_{j\in \ell}\|\tilde{\lambda}_j\|^2 + 2\sum_{i=0}^{\kappa-2}\sum_{j=i+1}^{\kappa-1}\langle \tilde{\lambda}_{l_j}, \tilde{\lambda}_{l_i} \rangle \\
&= \kappa\sum_{j\in \ell}\|\tilde{\lambda}_j\|^2 -
\sum_{i=0}^{\kappa-2}\sum_{j=i+1}^{\kappa-1}\| \tilde{\lambda}_{l_j}- \tilde{\lambda}_{l_i}\|^2.
\end{align*}
\rspace
\end{IEEEproof}

\begin{lemma}\label{lem:innerprod1}
For $1 \leq \kappa < n$ and any $\ell\in \mathcal{L}^{(n,\kappa)}$ we have
\begin{equation*}
\begin{split}
&2\left\langle \sum_{j\in \ell}\tilde{\lambda}_j , \sum_{i=0}^{n-1}\bar{\gamma}(\mathcal{L}_i^{(n,\kappa)})\tilde{\lambda}_i \right\rangle =
\bar{\gamma}(\mathcal{L}^{(n,\kappa)})\kappa\sum_{j\in
  l}\|\tilde{\lambda}_j\|^2 \\
&+ \kappa\sum_{i=0}^{n-1}\bar{\gamma}(\mathcal{L}_i^{(n,\kappa)})\|\tilde{\lambda}_i\|^2  - \sum_{j\in l}\sum_{i=0}^{n-1}\bar{\gamma}(\mathcal{L}_i^{(n,\kappa)})\|\tilde{\lambda}_j - \tilde{\lambda}_i\|^2.
\end{split}
\end{equation*}
\end{lemma}
\begin{IEEEproof}
{\allowdisplaybreaks
\begin{align*}
2&\left\langle \sum_{j\in \ell}\tilde{\lambda}_j , \sum_{i=0}^{n-1}\bar{\gamma}(\mathcal{L}_i^{(n,\kappa)})\tilde{\lambda}_i \right\rangle =
2\sum_{j\in \ell}\sum_{i=0}^{n-1}\bar{\gamma}(\mathcal{L}_i^{(n,\kappa)}) \langle \tilde{\lambda}_j ,\tilde{\lambda}_i \rangle   \\
&= -\sum_{j\in
  \ell}\sum_{i=0}^{n-1}\bar{\gamma}(\mathcal{L}_i^{(n,\kappa)})\|\tilde{\lambda}_j-\tilde{\lambda}_i\|^2
\\
&\quad
+ \sum_{j\in \ell}\sum_{i=0}^{n-1}\bar{\gamma}(\mathcal{L}_i^{(n,\kappa)})\left( \|\tilde{\lambda}_j\|^2 + \|\tilde{\lambda}_i\|^2\right) \\
&=
-\sum_{j\in \ell}\sum_{i=0}^{n-1}\bar{\gamma}(\mathcal{L}_i^{(n,\kappa)})\|\tilde{\lambda}_j-\tilde{\lambda}_i\|^2 \\
&\quad+ \kappa \bar{\gamma}(\mathcal{L}^{(n,\kappa)})\sum_{j\in \ell}\|\tilde{\lambda}_j\|^2  
+ \kappa\sum_{i=0}^{n-1}\bar{\gamma}(\mathcal{L}_i^{(n,\kappa)})\|\tilde{\lambda}_i\|^2
\end{align*}}
where the last equality follows from Lemma~\ref{lem:sumLj}.
\end{IEEEproof}

\begin{lemma}\label{lem:squard_sum_equiv}
For any $1\leq \kappa < n$ and $\ell\in \mathcal{L}^{(n,\kappa)}$ we have
\begin{equation*}
\begin{split}
&\bigg\|\frac{1}{\kappa}\sum_{j\in \ell}\tilde{\lambda}_j -
\frac{1}{\bar{\gamma}(\mathcal{L}^{(n,\kappa)})\kappa}\sum_{i=0}^{n-1}
\bar{\gamma}(\mathcal{L}_i^{(n,\kappa)})\tilde{\lambda}_i\bigg\|^2
\\
& =
\frac{1}{\bar{\gamma}(\mathcal{L}^{(n,\kappa)})^2\kappa^2}
\bigg(
\bar{\gamma}(\mathcal{L}^{(n,\kappa)})\sum_{j\in \ell}\sum_{i=0}^{n-1} \bar{\gamma}(\mathcal{L}_i^{(n,\kappa)})\|\tilde{\lambda}_j - \tilde{\lambda}_i\|^2 \\
&\quad -\bar{\gamma}(\mathcal{L}^{(n,\kappa)})^2\sum_{i=0}^{\kappa-2}\sum_{j=i+1}^{\kappa-1}\|\tilde{\lambda}_{l_i}-\tilde{\lambda}_{l_j}\|^2 \\
&\quad-\sum_{i=0}^{n-2}\sum_{j=i+1}^{n-1}\bar{\gamma}(\mathcal{L}_i^{(n,\kappa)})\bar{\gamma}(\mathcal{L}_j^{(n,\kappa)})\|\tilde{\lambda}_i - \tilde{\lambda}_j\|^2
\bigg).
\end{split}
\end{equation*}
\end{lemma}

\begin{IEEEproof}
We have that
\begin{equation*}
\begin{split}
&\left\|\frac{1}{\kappa}\sum_{j\in \ell}\tilde{\lambda}_j -
  \frac{1}{\kappa
    \bar{\gamma}(\mathcal{L}^{(n,\kappa)})}\sum_{i=0}^{n-1}
  \bar{\gamma}(\mathcal{L}_i^{(n,\kappa)})\tilde{\lambda}_i\right\|^2\\
&
= \frac{1}{\bar{\gamma}(\mathcal{L}^{(n,\kappa)})^2\kappa^2}\bigg(\bar{\gamma}(\mathcal{L}^{(n,\kappa)})^2\left\|\sum_{j\in \ell}\tilde{\lambda}_j\right\|^2
+ \left\|\sum_{i=0}^{n-1}\bar{\gamma}(\mathcal{L}_i^{(n,\kappa)})\tilde{\lambda}_i \right\|^2 \\
&\quad -
2\bar{\gamma}(\mathcal{L}^{(n,\kappa)})\left\langle \sum_{j\in \ell}\tilde{\lambda}_j, \sum_{i=0}^{n-1} \bar{\gamma}(\mathcal{L}_i^{(n,\kappa)}) \tilde{\lambda}_i \right\rangle \bigg)
\end{split}
\end{equation*}
which by use of Lemmas~\ref{lem:sumjside} and~\ref{lem:innerprod1} leads to
{\allowdisplaybreaks
\begin{align}\notag 
&\bigg\|\frac{1}{\kappa}\sum_{j\in \ell}\tilde{\lambda}_j -
\frac{1}{\bar{\gamma}(\mathcal{L}^{(n,\kappa)})\kappa}\sum_{i=0}^{n-1}
\bar{\gamma}(\mathcal{L}_i^{(n,\kappa)})\tilde{\lambda}_i\bigg\|^2 \\ \label{eq:sideperm}
&
=\frac{1}{\bar{\gamma}(\mathcal{L}^{(n,\kappa)})^2\kappa^2}\bigg(\bar{\gamma}(\mathcal{L}^{(n,\kappa)})^2\kappa\sum_{j\in \ell}\|\tilde{\lambda}_j\|^2 \\ \notag
&\quad
-
\bar{\gamma}(\mathcal{L}^{(n,\kappa)})^2\sum_{i=0}^{\kappa-2}\sum_{j=i+1}^{\kappa-1}\|\tilde{\lambda}_{l_i}-\tilde{\lambda}_{l_j}\|^2
\\ \notag
&\quad + \bar{\gamma}(\mathcal{L}^{(n,\kappa)})\kappa\sum_{i=0}^{n-1} \bar{\gamma}(\mathcal{L}_i^{(n,\kappa)})\|\tilde{\lambda}_i\|^2 \\ \notag
&\quad
-\sum_{i=0}^{n-2}\sum_{j=i+1}^{n-1}\bar{\gamma}(\mathcal{L}_i^{(n,\kappa)})\bar{\gamma}(\mathcal{L}_j^{(n,\kappa)})\|\tilde{\lambda}_i
- \tilde{\lambda}_j\|^2  \\ \notag
&\quad -\bar{\gamma}(\mathcal{L}^{(n,\kappa)})^2\kappa\sum_{j\in \ell}\|\tilde{\lambda}_j\|^2 - \bar{\gamma}(\mathcal{L}^{(n,\kappa)})\kappa\sum_{i=0}^{n-1}\bar{\gamma}(\mathcal{L}_i^{(n,\kappa)})\|\tilde{\lambda}_i\|^2 \\ \notag
&\quad  + \bar{\gamma}(\mathcal{L}^{(n,\kappa)})\sum_{j\in \ell}\sum_{i=0}^{n-1}\bar{\gamma}(\mathcal{L}_i^{(n,\kappa)})\|\tilde{\lambda}_j - \tilde{\lambda}_i\|^2 \bigg) \\ \notag
&= \frac{1}{\bar{\gamma}(\mathcal{L}^{(n,\kappa)})^2\kappa^2}\bigg(
\bar{\gamma}(\mathcal{L}^{(n,\kappa)})\sum_{j\in \ell}\sum_{i=0}^{n-1} \bar{\gamma}(\mathcal{L}_i^{(n,\kappa)})\|\tilde{\lambda}_j - \tilde{\lambda}_i\|^2 \\
&\quad -\bar{\gamma}(\mathcal{L}^{(n,\kappa)})^2\sum_{i=0}^{\kappa-2}\sum_{j=i+1}^{\kappa-1}\|\tilde{\lambda}_{l_i}-\tilde{\lambda}_{l_j}\|^2 \\ \notag
&\quad-\sum_{i=0}^{n-2}\sum_{j=i+1}^{n-1}\bar{\gamma}(\mathcal{L}_i^{(n,\kappa)})\bar{\gamma}(\mathcal{L}_j^{(n,\kappa)})\|\tilde{\lambda}_i - \tilde{\lambda}_j\|^2  \bigg).
\end{align}
}
\end{IEEEproof}

In order to establish step 3, we extend the proof technique previously used to find $\psi_{n,L}$ in~\cite{ostergaard:2004b}. 
Let $a_m$ denote the number of $\lambda_1$ points at distance $m$ from some $\lambda_0\in V_\pi(0)$.
It follows that a fixed $\lambda_0\in V_\pi(0)$ is paired with $a_m$ distinct $\lambda_1$ points when forming the $N_\pi$ $n$-tuples. 
Furthermore, let $b_m$ denote the number of $\lambda_2$ points which are paired with a fixed $(\lambda_0,\lambda_1)$ tuple. The total number of $n$-tuples (having $\lambda_0$ as first element) is given by $\sum_{m=1}^{r}a_m b_m$ where $r$ is the radius of $\tilde{V}$. It was shown in~\cite{ostergaard:2004b} that this procedure is asymptotically exact for large index values (and thereby large $r$). 

For a given $\lambda_0\in V_\pi(0)$, we seek to find an expression for $\sum_{\lambda_j \in \mathcal{T}_j(\lambda_i)}\|\lambda_i - \lambda_j\|^2$ where $\mathcal{T}_j(\lambda_i)$ was previously defined in the proof of Theorem~\ref{theo:separable} to be the set of $\lambda_j\in \Lambda_j$  which is in $n$-tuples having the specific $\lambda_i$ as the $i$th element.

\begin{proposition}\label{prop:squared_distances}
For $n=3$, any $1\leq L\in\mathbb{N}$, and asymptotically as $N_i\rightarrow\infty,\nu_i\rightarrow 0, \forall i$, we have
\begin{equation}\label{eq:ambmm2}
\begin{split}
&\sum_{\lambda_0 \in V_\pi(0)}\sum_{\lambda_j \in  \mathcal{T}_j(\lambda_i)}\|\lambda_i - \lambda_j\|^2 \\
&\qquad= \frac{L+2}L G(S_L) \psi_{3,L}^{2}\nu_c^{2/L} N_\pi^{1/L} N_\pi \frac{\tilde{\beta}_L}{\beta_L} 
\end{split}
\end{equation}
where $\tilde{\beta}_L$ is given by~(\ref{eq:tildebetaL}), $\beta_L$ by~(\ref{eq:betaL}) and $\psi_{3,L}$ by~(\ref{eq:psi3L}).
\end{proposition}
\begin{IEEEproof}
Without loss of generality, we let $\lambda_0=0$ so that 
\begin{align*}
\sum_{\lambda_j \in \mathcal{T}_j(\lambda_i)}\|\lambda_i - \lambda_j\|^2 &= 
\sum_{\lambda_j \in \mathcal{T}_j(0)}\|\lambda_j\|^2 \\
&= \frac{1}L\sum_{m=1}^r a_m b_m m^2
\end{align*}
where we used the fact that $m^2=\|\lambda_j\|^2/L$. 

The first part of the proof follows now by results of~\cite{ostergaard:2004b}. Specifically, from (65) in~\cite{ostergaard:2004b} (see also \cite[(H.43)]{ostergaard:2007a}) it follows that
\begin{equation}\label{eq:abm}
\frac{1}L\sum_{m=1}^r a_m b_m = 2\frac{\omega_L\omega_{L-1}}{\nu_1\nu_2}\frac{1}{L+1}\tilde{\beta}_Lr^{2L}
\end{equation}
where it is easy to show that we can replace $a_mb_m$ by $a_m b_m m^2$ and obtain
{\allowdisplaybreaks
\begin{align*}
\frac{1}L\sum_{m=1}^r a_m b_m m^2 &= 2\frac{\omega_L\omega_{L-1}}{\nu_1\nu_2}\frac{1}{L+1}\tilde{\beta}_Lr^{2L+2} \\
&\overset{(a)}{=} 2\frac{\omega_L\omega_{L-1}}{\nu_1\nu_2}\frac{1}{L+1}\tilde{\beta}_L 
\tilde{\nu}^{2/L} \tilde{\nu}^2 \frac{1}{\omega_L^{2+2/L}} \\
&\overset{(b)}{=}\frac{1}{L}\frac{1}{\omega_L^{2/L}} \tilde{\nu}^{2/L} N_0 \frac{\tilde{\beta}_L}{\beta_L} \\
&\overset{(c)}{=} \frac{L+2}{L} G(S_L) \tilde{\nu}^{2/L} N_0 \frac{\tilde{\beta}_L}{\beta_L} \\
&= \frac{L+2}L G(S_L) \psi_{n,L}^{2}\nu_c^{2/L} N_\pi^{1/L} N_0 \frac{\tilde{\beta}_L}{\beta_L} 
\end{align*}}
where $(a)$ follows by use of~(\ref{eq:nutilde}), i.e.\ $\tilde{\nu} = \omega_L r^L = \psi_{n,L}^L \nu_c \sqrt{N_\pi}$ and $(b)$ follows by use of~(\ref{eq:psi3L}). Finally, $(c)$ follows since $\omega_L^{-2/L} = (L+2)G(S_L)$.

The proof now follows by using the fact that~(\ref{eq:abm}) is independent of $\lambda_0$ so that, since there are $N_\pi/N_0$ distinct $\lambda_0$'s in $V_\pi(0)$, we get 
\begin{equation*}
\begin{split}
&\frac{1}L\sum_{\lambda_0\in V_\pi(0)}\sum_{m=1}^r a_m b_m m^2 \\
&\qquad= \frac{L+2}L G(S_L) \psi_{n,L}^{2}\nu_c^{2/L} N_\pi^{1/L} N_\pi \frac{\tilde{\beta}_L}{\beta_L}.
\end{split}
\end{equation*}
\end{IEEEproof}

We are now in a position to prove the theorem.

\begin{IEEEproof}[Proof of Theorem \ref{theo:dist_3desc}]
Let $\tilde{\lambda}_i=\mu_i\alpha_i(\lambda_c)$, then asymptotically as $N_i\rightarrow\infty, \nu_i\rightarrow 0, \forall i$, we have that
{\allowdisplaybreaks
\begin{align*}
&\frac{1}{N_\pi}\sum_{\lambda_c\in V_\pi(0)}\left\|\lambda_c - \frac{1}{\kappa}\sum_{i\in\ell}\tilde{\lambda}_i\right\|^2 \\
&\overset{(a)}{=}\frac{1}{N_\pi}
\sum_{\lambda_c\in V_\pi(0)}\left\| \frac{1}{\kappa}\sum_{j\in \ell}\tilde{\lambda}_j - \frac{1}{\bar{\gamma}(\mathcal{L}^{(n,\kappa)})\kappa}\sum_{i=0}^{n-1} \bar{\gamma}(\mathcal{L}_i^{(n,\kappa)})\tilde{\lambda}_i\right\|^2 \\
&\overset{(b)}{=} 
\frac{1}{N_\pi}\sum_{\lambda_c\in V_\pi(0)}\frac{1}{\bar{\gamma}(\mathcal{L}^{(n,\kappa)})^2\kappa^2}
\bigg(\\
&\qquad
\bar{\gamma}(\mathcal{L}^{(n,\kappa)})\sum_{j\in \ell}\sum_{i=0}^{n-1} \bar{\gamma}(\mathcal{L}_i^{(n,\kappa)})\|\tilde{\lambda}_j - \tilde{\lambda}_i\|^2 \\
&\quad
-\bar{\gamma}(\mathcal{L}^{(n,\kappa)})^2\sum_{i=0}^{\kappa-2}\sum_{j=i+1}^{\kappa-1}\|\tilde{\lambda}_{l_i}-\tilde{\lambda}_{l_j}\|^2
\\
&\quad
-\sum_{i=0}^{n-2}\sum_{j=i+1}^{n-1}\bar{\gamma}(\mathcal{L}_i^{(n,\kappa)})\bar{\gamma}(\mathcal{L}_j^{(n,\kappa)})\|\tilde{\lambda}_i - \tilde{\lambda}_j\|^2
\bigg) \\
&\overset{(c)}{=} 
\frac{1}{\bar{\gamma}(\mathcal{L}^{(n,\kappa)})^2\kappa^2}
\bigg(
\bar{\gamma}(\mathcal{L}^{(n,\kappa)})\sum_{j\in
  \ell}\sum_{\substack{i=0\\ i\neq j}}^{n-1}
\bar{\gamma}(\mathcal{L}_i^{(n,\kappa)}) \\
&\quad
-\bar{\gamma}(\mathcal{L}^{(n,\kappa)})^2 \binom{\kappa}{2} 
-\sum_{i=0}^{n-2}\sum_{j=i+1}^{n-1}\bar{\gamma}(\mathcal{L}_i^{(n,\kappa)})\bar{\gamma}(\mathcal{L}_j^{(n,\kappa)})
\bigg) \\
&\quad\times
\frac{L+2}L G(S_L) \psi_{3,L}^{2}\nu_c^{2/L} N_\pi^{1/L} \frac{\tilde{\beta}_L}{\beta_L} \\
&= 
\frac{1}{\bar{\gamma}(\mathcal{L}^{(n,\kappa)})^2\kappa^2}
\bigg(
\kappa^2\bar{\gamma}(\mathcal{L}^{(n,\kappa)})^2 - \bar{\gamma}(\mathcal{L}^{(n,\kappa)})\sum_{j\in \ell}\bar{\gamma}(\mathcal{L}_j^{(n,\kappa)}) \\
&-\bar{\gamma}(\mathcal{L}^{(n,\kappa)})^2 \binom{\kappa}{2} -\sum_{i=0}^{n-2}\sum_{j=i+1}^{n-1}\bar{\gamma}(\mathcal{L}_i^{(n,\kappa)})\bar{\gamma}(\mathcal{L}_j^{(n,\kappa)})
\bigg) \\
&\quad\times\frac{L+2}L G(S_L) \psi_{3,L}^{2}\nu_c^{2/L} N_\pi^{1/L} \frac{\tilde{\beta}_L}{\beta_L} 
\end{align*}
where $(a)$ follows by Proposition~\ref{prop:lclmean}, $(b)$ follows by Lemma~\ref{lem:squard_sum_equiv}, and $(c)$ follows by Proposition~\ref{prop:squared_distances}. The proof now follows by observing that the second term of~(\ref{eq:barDL_1}) is negligible compared to the first term of~(\ref{eq:barDL_1}).}
\end{IEEEproof}

\section{Proof of Lemma~\ref{lem:scec}}\label{app:scec}
We consider a zero-mean and unit-variance Gaussian source $X$ and define three random variables $Y_i \triangleq X + Q_i, i=0,1,2$ where the $Q_i$'s are identically distributed jointly Gaussian random variables (independent of $X$) with variance $\sigma_q^2$ and covariance matrix $Q$ given by 
\begin{equation*}
Q =  \sigma_q^2 
\begin{bmatrix}
1 & \rho & \rho \\
\rho & 1 & \rho \\
\rho & \rho & 1 
\end{bmatrix}
\end{equation*}
where the correlation coefficient satisfies $-\frac{1}2 < \rho \leq \frac{1}{2}$. 
It is easy to show that the MMSE when estimating $X$ from any set of $m$ $Y_i$'s is given by~\cite{pradhan:2004,venkataramani:2003} 
\begin{equation*}
\text{MMSE}_m = \frac{ \sigma_q^2( 1+ (m-1)\rho) }{m + \sigma_q^2(1+(m-1)\rho)}.
\end{equation*}
In the high-resolution case where $\sigma_q^2\ll 1$ it follows that we have
\begin{equation*}
\text{MMSE}_1 = \sigma_q^2
\end{equation*}
\begin{equation*}
\text{MMSE}_2 = \frac{1}{2}\sigma_q^2(1+\rho)
\end{equation*}
and
\begin{equation*}
\text{MMSE}_3 = \frac{1}{3}\sigma_q^2(1+2\rho).
\end{equation*}
It was shown in~\cite{pradhan:2004} that, the description rate $R$ is given by
\begin{equation*}
R = \frac{1}{2}\log_2\left( \frac{1+\sigma_q^2}{\sigma_q^2(1-\rho)}\right) \left( \frac{1-\rho}{1+2\rho}\right)^{\frac{1}{3}}
\end{equation*}
so that 
\begin{align*}
\sigma_q^2 &= \left( (1-\rho)2^{2R} \left( \frac{ 1+2\rho}{1-\rho} \right)^{\frac{1}3} - 1\right)^{-1} \\
&\approx (1-\rho)^{-\frac{2}3}(1+2\rho)^{-\frac{1}3}2^{-2R}
\end{align*}
where the approximation becomes exact at high resolution (i.e.\ for $R\gg 1$).
We can now form the high-resolution distortion product 
\begin{align*}
D^{\pi} &= \frac{\sigma_q^6}{6}(1+\rho)(1+2\rho) \\
&= \frac{1}6(1+\rho)(1-\rho)^{-2}2^{-6R} \\
&= \frac{1}{27} 2^{-6R}
\end{align*}
where the last inequality follows by inserting $\rho\rightarrow -\frac{1}{2}$ which corresponds to having a high side-to-central distortion ratio, i.e.\ it resembles the asymptotical condition of letting $N_i\rightarrow\infty$ in the IA based approach. This proves the lemma. \hfill\IEEEQEDclosed

\section*{Acknowledgment}
The authors are extremely grateful to the referees and the associate editor for providing 
numerous and invaluable critical comments and suggestions, which substantially improved the quality and the presentation of the manuscript.

\bibliographystyle{ieeetr}


\begin{thebibliography}{10}

\bibitem{ozarow:1980}
L.~Ozarow, ``On a source-coding problem with two channels and three
  receivers,'' {\em Bell System Technical Journal}, vol.~59, pp.~1909 -- 1921,
  December 1980.

\bibitem{elgamal:1982}
A.~A.~E. Gamal and T.~M. Cover, ``Achievable rates for multiple descriptions,''
  {\em IEEE Trans. Inf. Theory}, vol.~IT-28, pp.~851 -- 857, November 1982.

\bibitem{zamir:1999}
R.~Zamir, ``Gaussian codes and shannon bounds for multiple descriptions,'' {\em
  IEEE Trans. Inf. Theory}, vol.~45, pp.~2629 -- 2636, November 1999.

\bibitem{zamir:2000}
R.~Zamir, ``Shannon type bounds for multiple descriptions of a stationary
  source,'' {\em Journal of Combinatorics, Information and System Sciences},
  pp.~1 -- 15, December 2000.

\bibitem{chen:2007}
J.~Chen, C.~Tian, and S.~Diggavi, ``Multiple description coding for stationary
  and ergodic sources,'' in {\em Proc. Data Compression Conf.}, pp.~73 -- 82,
  March 2007.

\bibitem{venkataramani:2003}
R.~Venkataramani, G.~Kramer, and V.~K. Goyal, ``Multiple description coding
  with many channels,'' {\em IEEE Trans. Inf. Theory}, vol.~49, pp.~2106 --
  2114, September 2003.

\bibitem{pradhan:2004}
S.~S. Pradhan, R.~Puri, and K.~Ramchandran, ``{$n$}-channel symmetric multiple
  descriptions--part {I}: {$(n,k)$} source-channel erasure codes,'' {\em IEEE
  Trans. Inf. Theory}, vol.~50, pp.~47 -- 61, January 2004.

\bibitem{puri:2005}
R.~Puri, S.~S. Pradhan, and K.~Ramchandran, ``{$n$}-channel symmetric multiple
  descriptions- part {II}: An achievable rate-distortion region,'' {\em IEEE
  Trans. Inf. Theory}, vol.~51, pp.~1377 -- 1392, April 2005.

\bibitem{wang:2006}
H.~Wang and P.~Viswanath, ``Vector gaussian multiple description with
  individual and central receivers,'' in {\em Proc. IEEE Int. Symp. Information
  Theory}, 2006.

\bibitem{wang:2008}
H.~Wang and P.~Wiswanath, ``Vector {G}aussian multiple description with two
  levels of receivers,'' {\em IEEE Trans. Inf. Theory}, 2008.
\newblock Submitted.

\bibitem{tian:2008b}
C.~Tian, S.~Mohajer, and S.~Diggavi, ``Approximating the {G}aussian multiple
  description rate region under symmetric distortion constraints,'' in {\em
  Proc. IEEE Int. Symp. Information Theory}, (Toronto, Canada), pp.~1413 --
  1417, July 2008.

\bibitem{vaishampayan:2001}
V.~A. Vaishampayan, N.~J.~A. Sloane, and S.~D. Servetto, ``Multiple-description
  vector quantization with lattice codebooks: Design and analysis,'' {\em IEEE
  Trans. Inf. Theory}, vol.~47, pp.~1718 -- 1734, July 2001.

\bibitem{ostergaard:2004b}
J.~{\O}stergaard, J.~Jensen, and R.~Heusdens, ``{$n$}-channel
  entropy-constrained multiple-description lattice vector quantization,'' {\em
  IEEE Trans. Inf. Theory}, vol.~52, pp.~1956 -- 1973, May 2006.

\bibitem{diggavi:2002}
S.~N. Diggavi, N.~J.~A. Sloane, and V.~A. Vaishampayan, ``Asymmetric multiple
  description lattice vector quantizers,'' {\em IEEE Trans. Inf. Theory},
  vol.~48, pp.~174 -- 191, January 2002.

\bibitem{vaishampayan:1993}
V.~A. Vaishampayan, ``Design of multiple description scalar quantizers,'' {\em
  IEEE Trans. Inf. Theory}, vol.~39, pp.~821 -- 834, May 1993.

\bibitem{chen:2006}
J.~Chen, C.~Tian, T.~Berger, and S.~S. Hemami, ``Multiple description
  quantization via {G}ram-{S}chmidt orthogonalization,'' {\em IEEE Trans. Inf.
  Theory}, vol.~52, pp.~5197 -- 5217, December 2006.

\bibitem{ostergaard:2009}
J.~{\O}stergaard and R.~Zamir, ``Multiple-description coding by dithered
  delta-sigma quantization,'' {\em IEEE Trans. Inf. Theory}, vol.~55, pp.~4661
  -- 4675, October 2009.

\bibitem{kochman:2008}
Y.~Kochman, J.~{\O}stergaard, and R.~Zamir, ``Noise-shaped predictive coding
  for multiple descriptions of a colored {G}aussian source,'' in {\em Proc.
  Data Compression Conf.}, pp.~362 -- 371, March 2008.

\bibitem{mohajer:2008a}
S.~Mohajer, C.~Tian, and S.~N. Diggavi, ``Asymmetric {G}aussian multiple
  descriptions and asymmetric multilevel diversity coding,'' in {\em Proc. IEEE
  Int. Symp. Information Theory}, pp.~1992 -- 1996, July 2008.

\bibitem{tian:2009a}
C.~Tian, S.~Mohajer, and S.~N. Diggavi, ``Approximating the {G}aussian multiple
  description rate region under symmetric distortion constraints,'' {\em IEEE
  Trans. Inf. Theory}, vol.~55, pp.~3869 -- 3891, August 2009.

\bibitem{conway:1999}
J.~H. Conway and N.~J.~A. Sloane, {\em Sphere packings, Lattices and Groups}.
\newblock Springer, 3rd~ed., 1999.

\bibitem{fleming:2004}
M.~Fleming, Q.~Zhao, and M.~Effros, ``Network vector quantization,'' {\em IEEE
  Trans. Inf. Theory}, vol.~50, pp.~1584 -- 1604, August 2004.

\bibitem{berger-wolf:2002}
T.~Y. Berger-Wolf and E.~M. Reingold, ``Index assignment for multichannel
  communication under failure,'' {\em IEEE Trans. Inf. Theory}, vol.~48,
  pp.~2656 -- 2668, October 2002.

\bibitem{tian:2004}
C.~Tian and S.~S. Hemami, ``Universal multiple description scalar quantization:
  analysis and design,'' {\em IEEE Trans. Inf. Theory}, vol.~50, pp.~2089 --
  2102, September 2004.

\bibitem{orchard:1997}
M.~T. Orchard, Y.~Wang, V.~Vaishampayan, and A.~R. Reibman, ``Redundancy
  rate-distortion analysis of multiple description coding using pairwise
  correlating transforms,'' in {\em Proc. IEEE Conf. on Image Proc.}, vol.~1,
  pp.~608 -- 611, 1997.

\bibitem{goyal:2001}
V.~K. Goyal and J.~Kova\v{c}evi\'{c}, ``Generalized multiple descriptions
  coding with correlating transforms,'' {\em IEEE Trans. Inf. Theory}, vol.~47,
  pp.~2199 -- 2224, September 2001.

\bibitem{balan:2000}
R.~Balan, I.~Daubechies, and V.~Vaishampayan, ``The analysis and design of
  windowed fourier frame based multiple description source coding schemes,''
  {\em IEEE Trans. Inf. Theory}, vol.~46, pp.~2491 -- 2536, November 2000.

\bibitem{chou:1999}
P.~A. Chou, S.~Mehrotra, and A.~Wang, ``Multiple description decoding of
  overcomplete expansions using projections onto convex sets,'' in {\em Proc.
  Data Compression Conf.}, pp.~72 -- 81, March 1999.

\bibitem{zamir:2002}
R.~Zamir, S.~Shamai, and U.~Erez, ``Nested linear/lattice codes for structured
  multiterminal binning,'' {\em IEEE Trans. Inf. Theory}, vol.~Special A.D.
  Wyner issue, pp.~1250 -- 1276, June 2002.

\bibitem{zamir:1996}
R.~Zamir and M.~Feder, ``On lattice quantization noise,'' {\em IEEE Trans. Inf.
  Theory}, vol.~42, pp.~1152 -- 1159, July 1996.

\bibitem{erez:2004}
U.~Erez and R.~Zamir, ``Achieving 1/2 log(1+snr) on the awgn channel with
  lattice encoding decoding,'' {\em IEEE Trans. Inf. Theory}, vol.~50,
  pp.~2293--2314, October 2004.

\bibitem{krithivasan:2008}
D.~Krithivasan and S.~S. Pradhan, ``Lattices for distributed source coding:
  Jointly {G}aussian sources and reconstruction of a linear function,'' {\em
  IEEE Trans. Inf. Theory}, 2008.
\newblock \emph{Submitted to IEEE Trans.\ Inf.\ Theory.} Available on ArXiv.

\bibitem{gray:1990}
R.~M. Gray, {\em Source Coding Theory}.
\newblock Kluwer Academic Publishers, 1990.

\bibitem{mazo:1990}
J.~E. Mazo and A.~M. Odlyzko, ``Lattice points in high-dimensional spheres,''
  {\em Monatsh. Math.}, vol.~110, pp.~47 -- 61, 1990.

\bibitem{fricker:1982}
F.~Fricker, {\em Einf\"{u}hrung in die gitterpunktlehre}.
\newblock Birkh\"{a}user, 1982.

\bibitem{ostergaard:2007a}
J.~{\O}stergaard, {\em {M}ultiple-description lattice vector quantization}.
\newblock PhD thesis, Delft University of Technology, Delft, The Netherlands,
  June 2007.

\bibitem{zhang:2010}
G.~Zhang, J.~{\O}stergaard, J.~Klejsa, and W.~B. Kleijn, ``High-rate analysis
  of symmetric ${L}$-channel multiple description coding,'' {\em IEEE Trans. on
  Communications}, 2010.
\newblock Submitted.

\bibitem{vaishampayan:1998}
V.~A. Vaishampayan and J.-C. Batllo, ``Asymptotic analysis of multiple
  description quantizers,'' {\em IEEE Trans. Inf. Theory}, vol.~44, pp.~278 --
  284, January 1998.

\bibitem{tian:2008}
C.~Tian and J.~Chen, ``A novel coding scheme for symmetric multiple description
  coding,'' in {\em Proc. IEEE Int. Symp. Information Theory}, (Toronto,
  Canada), pp.~1418 -- 1422, July 2008.

\bibitem{buchholz:1992}
R.~H. Buchholz, ``Perfect pyramids,'' {\em Bulletin Australian Mathematical
  Society}, vol.~45, no.~3, 1992.

\end{thebibliography}

\begin{IEEEbiographynophoto}{Jan {\O}stergaard}
received the M.Sc. degree in electrical engineering from Aalborg University, Aalborg, Denmark, in 1999 and the Ph.D. degree (with cum laude) in electrical engineering from Delft University of Technology, Delft, The Netherlands, in 2007. From 1999 to 2002, he worked as an R\&D engineer at ETI A/S, Aalborg, Denmark, and from 2002 to 2003, he worked as an R\&D engineer at ETI Inc., Virginia, United States. Between September 2007 and June 2008, he worked as a post-doctoral researcher in the Centre for Complex Dynamic Systems and Control, School of Electrical Engineering and Computer Science, The University of Newcastle, NSW, Australia. He has also been a visiting researcher at Tel Aviv University, Tel Aviv, Israel, and at Universidad Technica Federica Santa Maria in Valparaiso, Chile.
Jan {\O}stergaard is currently a post-doctoral researcher at Aalborg University, Aalborg, Denmark. He has received a Danish Independent Research Council's Young Researcher's Award and a fellowship from the Danish Research Council for Technology and Production Sciences.
His current research interest include rate-distortion theory, joint source-channel coding, multiple-description coding, distributed source coding, lattice vector quantization, speech, audio, image, and video coding. 
\end{IEEEbiographynophoto}

\newpage

\begin{IEEEbiographynophoto}{Richard Heusdens}
received his M.Sc. and Ph.D. degree from the Delft 
University of Technology, the Netherlands, in 1992 and 1997, respectively. In the spring of 1992 he joined the digital signal processing group at the Philips Research Laboratories in Eindhoven, the Netherlands, where he worked on various topics in the field of signal processing, such as image/video/audio/speech compression, and VLSI architectures for image-processing algorithms. In 1997 he joined the Circuits and Systems Group of the Delft University of Technology, where he was a postdoctoral researcher. In 2000 he moved to the Information and Communication Theory (ICT) Group where he became an assistant professor, responsible for the audio and speech processing activities within the ICT group. Since 2002 he is an associate professor. He held visiting positions at KTH (Royal Institute of Technology) and Philips Research Laboratories. He is involved in research projects that cover subjects such as audio and speech coding, speech enhancement, digital watermarking of audio and acoustical echo cancellation. He is the associate editor of the EURASIP Journal of Applied Signal Processing and the EURASIP Journal on Audio, Speech and Music Processing.
\end{IEEEbiographynophoto}

\begin{IEEEbiographynophoto}{Jesper Jensen}
received the M.Sc degree in electrical engineering and the
Ph.D degree in signal processing from Aalborg University, Aalborg,
Denmark, in 1996 and 2000, respectively.
From 1996 to 2000 he was with the Center for Person Kommunikation (CPK),
Aalborg University, as a Ph.D student and assistant research professor.
From 2000 to 2007 he was a post-doctoral researcher and assistant
professor with Delft University of Technology, The Netherlands, and an
external associate professor with Aalborg University, Denmark.
Currently, he is with Oticon A/S, Denmark. His main research interests
are in the area of acoustical signal processing, including signal
retrieval from noisy observations, coding, speech and audio modification
and synthesis, intelligibility enhancement of speech signals, and
perceptual aspects of signal processing.
\end{IEEEbiographynophoto}

\end{document}